\documentclass[12pt,a4paper, left=0in, right=0in, top=0in, bottom=0in, margin=0.5in]{article}

\usepackage{geometry}
\usepackage{titlesec}
\usepackage{setspace}
\usepackage{lipsum}
\usepackage{natbib}
\usepackage[ruled]{algorithm2e}
\usepackage{amssymb}
\usepackage{amsmath}
\usepackage{amsthm}
\usepackage{array}
\usepackage{float}
\usepackage{mdframed}
\usepackage{appendix}
\usepackage{multirow}
\usepackage{graphicx}
\usepackage{mdframed}
\usepackage{soul}
\usepackage[hyperindex,breaklinks]{hyperref}
\hypersetup{colorlinks = true, urlcolor = green, linkcolor = blue,
  citecolor = blue }
\usepackage[bottom]{footmisc}

\usepackage{xcolor}

\newtheorem{definition}{Definition}
\newtheorem{example}{Example}

\newtheorem{hypothesis}{Hypothesis}
\newtheorem{theorem}{Theorem}
  {\par\noindent\hspace*{0.5em}\begin{minipage}{\textwidth}\ignorespaces}%
  {\end{minipage}\par}

\newtheorem{proposition}{Proposition}
\newtheorem{remark}{Remark}

\theoremstyle{definition}
\newtheorem{proced}{Procedure}

\title{\bfseries Quasi-randomization tests for network interference: a random graph approach}
\author{Supriya Tiwari\footnote{Email: supriya\textunderscore tiwari@isb.edu \\ 
}
\hspace{0.3cm}and\hspace{0.15cm}Pallavi Basu\\
  \selectfont Indian School of Business, Hyderabad 500111, India}
\date{\today}

\begin{document}

\maketitle

\abstract{ \noindent

  Network interference occurs when the treatment status of one unit affects the potential outcomes of other units, giving rise to spillover effects that are difficult to test for. We propose treating the network as a random variable rather than a fixed quantity to address this challenge. This overcomes a key challenge of non-imputability of potential outcomes under the null and avoids the computational intractability of existing conditional randomization tests. Our quasi-randomization test builds the null distribution of no spillover effects using random graph null models, is exactly valid in finite samples under mild assumptions on the network-generating process, and offers substantially improved power over existing methods, particularly in cluster-randomized trials. We validate our approach via simulation and illustrate it by testing for interference in a weather insurance adoption experiment in rural China.

\begin{quote}
\end{quote}



  

\section{Introduction}

A common assumption in causal inference is the Stable Unit Treatment Value Assumption (SUTVA), which requires that the treatment given to a unit does not affect the potential outcomes of other units in the population \citep{cox1958planning}. While this allows for unbiased estimation of treatment effects and valid inferential strategies, many experimental settings involve \textit{interference} between units through network ties. In such settings, the treatment status of one unit can affect the outcomes of neighboring units through the network, giving rise to \textit{spillover effects}.\footnote{For example, experiments on social media networks such as LinkedIn or online marketplaces such as eBay observe interference owing to the interconnected nature of these platforms \citep{pouget2019testing, blake2014marketplace}.}
The total treatment effect can in these cases be decomposed into a direct effect and a spillover effect \citep{sobel2006randomized}. Throughout this paper, we restrict attention to interference operating through a network structure, and we assume the network is formed prior to and independently of treatment assignment; in particular, treatment does not alter link formation or interaction patterns among units.

\cite{savje2021average} shows that treatment effect estimators are robust to moderate interference, though inference is not accurate. The classical method for testing treatment effects is the Fisher Randomization Test (FRT; \citep{fisher1936design}, \citep{zhang2023randomization}), which tests a sharp null hypothesis of no treatment effect by exploiting the counterfactual imputability under the null. Testing under the null of no spillover effects creates imputability issues, as counterfactuals where direct treatment status is varied cannot be imputed under the null. This is referred to as a \textit{non-sharp} null, and determining the randomization distribution of the test statistic becomes challenging. A growing literature addresses this by restricting attention to a chosen subset of units, following \cite{aronow2012general}, known as \textit{focal} units. These procedures fix the treatment assignment for the focal units, creating a sharp null for these units. This yields a conditional randomization test with valid exact p-values \citep{athey2018exact, basse2019randomization, puelz2022graph}. \cite{athey2018exact} formulates this approach to a given network and evaluates a range of non-sharp null hypotheses. \cite{basse2019randomization} provides a framework for general forms of interference. \cite{puelz2022graph} presents a graph-theoretic approach to constructing the conditioning mechanism. 



In this paper, we depart from this literature by treating the network itself as a random variable, rather than a fixed quantity. This opens up a new path for interference testing. The idea draws a parallel to propensity score methods in observational studies: just as the probability of treatment assignment can be modeled as a function of covariates, the propensity for two nodes to be connected can be modeled as a function of network characteristics. 
Current methods that take the network as fixed are valid in finite samples but can have low power, as our numerical studies show. We demonstrate that leveraging information from the network-generating process yields substantial power gains. Consider cluster-randomized trials as an example: because the network-formation process is a precursor to cluster formation, a realistic model of that process is especially informative in such trials. More broadly, the gains apply wherever the exposures are strongly dependent, not only under cluster randomization. These gains come at the cost of an assumption on the network-generating process, a trade-off between potential model-misspecification bias and improved power that we argue is reasonable in many practical settings.

We note that this shift from fixed-network to random-network inference changes the inferential target. Standard design-based inference provides statements conditional on the realized network. Our approach instead averages over hypothetical networks generated by the same process as the observed network. This is appropriate in settings where the observed network is itself one realization of an ongoing social process, such as evolving friendship ties or communication patterns, rather than a fixed structural feature of the population. We discuss this interpretation explicitly in Section \ref{Setup}.


We assume the network data-generating process is characterized by a vector of sufficient statistics and can thus be represented as an exponential random graph model. We first treat the degree sequence as a sufficient statistic, following \cite{newman2001random}, where random graphs with arbitrary degree distributions capture heterogeneity in degree and closely represent real-world social networks. This enables us to build a quasi-randomization test by randomizing over graphs generated by this process, conditional on the observed treatment assignment. We then generalize to a vector of permutation-invariant sufficient statistics and present a method for a general exponential random graph model.


Model-based approaches with distributional assumptions about the outcome model have previously been considered for identifying and estimating peer effects \citep{manski1993identification, manski2013identification, bowers2013reasoning, toulis2013estimation, blume2015linear}. More recently, \cite{li2022random} and \cite{leung2020treatment} proposed estimation strategies for spillover effects under distributional assumptions on the network formation model, drawing on random graph asymptotics. Our work complements these estimation approaches by providing a finite-sample valid testing procedure within the random graph framework. Unlike the asymptotic confidence intervals of \cite{li2022random}, our method delivers exact p-values under mild assumptions on the network-generating process. Subsequent to our work, \cite{zhong2024unconditional} develops a related partial-sharp-null framework. Unlike \cite{zhong2024unconditional}, who randomizes over treatment assignments unconditionally, our procedure randomizes over network realizations while holding treatment fixed. Both \cite{zhong2024unconditional} and this paper test the null hypothesis across all units. While \cite{zhong2024unconditional} remains distribution-free, we propose the use of random-graph null models to boost the power of the methods. This is particularly evident in cluster-randomized trials, where the treatment assignment space is too restricted to support a non-degenerate null distribution. \cite{zhong2024unconditional} does not study the size of the imputable sets in this setting.

Our contributions are threefold. First, we embed the random graph interference model within the potential outcomes framework. We define sharp and non-sharp null hypotheses for random graph models and develop a randomization-based test that achieves finite-sample validity in line with Fisherian procedures.

Second, we address technical challenges in current methods. Randomization tests based on network-generating processes yield imputable test statistics under the null, conditional on treatment assignment. This overcomes a key limitation of existing approaches, which restrict attention to focal units and test the no-spillover null only for that subset. Our method tests the stricter null of no spillover effects across all units. Crucially, by randomizing over the network distribution rather than the treatment assignment, our method avoids the degeneracy that persists in existing approaches in cluster-randomized trials, where the treatment assignment space is too restricted to support valid inference.

Third, we establish the computational efficiency of the proposed tests. The permutation test based on degree sequence preserves the degree sequence of the observed graph and admits fast polynomial-time algorithms for large graphs \citep{viger2005efficient, fosdick2018configuring}. The generalization to exponential random graph models reduces to permuting within graph isomorphism classes, connecting to a broader literature on random graph null models in network analysis \citep{milenkovic2009optimized, sah2014exploring}.


To illustrate our method, we reanalyze a randomized experiment in rural China in which rice crop farmers are informed about a weather insurance product \citep{cai2015social}. The hypothesis is that farmers with a higher proportion of informed neighbors exhibit a higher adoption rate. We apply our methodology to a subset of the original experiment and find significant evidence of information spillover among farmers (p-value = 0.0126).


The remaining paper is structured as follows. In Section \ref{Setup}, we set up notation and formalize the null hypothesis of interest, incorporating the network random variable. We also describe the state-of-the-art methodology of using a conditioning mechanism to test a non-sharp null. Sections \ref{quasi-rand-test} and \ref{Experimental design} present a randomization test based on a given network-generating process and discuss implementation under different experimental designs. Sections \ref{quasi permutation test} and \ref{Permutation test} generalize this to exponential random graph models, introducing the permutation test. In Section \ref{simulation study}, we present extensive numerical simulations to evaluate the performance of the proposed test. Sections \ref{illustration} and \ref{conclusion} present an empirical application illustrating our methodology and conclude the paper. Proofs, algorithms, and additional numerical results are presented in the Appendix.

\section{Setup}
\label{Setup}
\subsection{Notation}

Consider a population $\mathbb{P} = \{1,2,..., N\}$ where $N$ is the total number of units in the population. Consider a vector $Z \in {\mathbb{Z}}^N$ where each element in the vector indicates treatment assigned to the corresponding indexed unit in the population. Here, $\mathbb{Z}$ represents the set of possible treatments to which the population units can be exposed. We take $\mathbb{Z}=\{0,1\}$ for this article. Hence, any treatment assignment vector is an N-dimensional binary vector. We consider the population $\mathbb{P}$ to be connected via an undirected graph $G = (\mathbb{P},E)$ where $E \subseteq \mathbb{P}\times\mathbb{P}$. For distinct units $i$ and $j$, $(i,j) \in E$ represents a connection between them. With some slight abuse of notation, we also represent the graph's adjacency matrix by $G$. Thus, $G$ is an $N \times N$ binary matrix. Here, for some $i$ and $j$, $G_{ii} = 0$ and $G_{ij} = G_{ji}$ owing to $G$ being an undirected graph with no self-loops. If $G_{ij} = 1$, then we say units $i$ and $j$ are neighbours. This can also be denoted by $i \in \mathcal{N}_{G}(j)$ or vice-versa. Here, $\mathcal{N}_{G}(i)$ represents all the neighbors of unit $i$ in graph $G$. We may write $\mathcal{N}(i)$ if the graph is clear in context. We define the total number of unit neighbors of $i$ as $ i$'s degree. This is denoted by $deg(i)$ and equals $|\mathcal{N}(i)|$. We denote the distance between two units $i$ and $j$ as $dist(i,j)$. It is defined as the shortest path length between the two units. By convention, we take $dist(i, i) = 0$ and $dist(i,j) = \infty$ if no path exists between the two distinct units. Consider $\mathbb{P'}\subseteq \mathbb{P}$. We denote $G[\mathbb{P'}]$ as the sub-graph induced by $\mathbb{P'}$. That is, $G[\mathbb{P'}]=(\mathbb{P'},E')$ where $E' = \{(i,j)\in\mathbb{P'}\times\mathbb{P'}|\:(i,j)\in E\}$.

Each unit $i$ in the population is `characterized' by a set of pre-treatment or covariate vectors given by $X_i$, where $X_i$ is an M-dimensional vector. $X_i$ is fixed and influences the outcome via the network's potential outcome function or formation. Let $X$ denote the $  N\times M$ matrix where each row corresponds to the respective unit's covariate vector. Consider the treatment vector to be drawn from a probability distribution $P_{X}(Z, G): \{0,1\}^N\times \{0,1\}^{N\times N} \rightarrow [0,1]$. We define the conditional probability distribution $P(Z|G): \{0,1\}^N \rightarrow [0,1] \subset \mathbb{R}$ as $P_{Z|G}$. Let $Z_{obs}$ be the realized treatment drawn from the treatment assignment distribution $P_{Z|G}$. Consider the graph $G$ to be coming from the marginal probability distribution $P(G): \{0,1\}^{N \times N} \rightarrow [0,1] \subset \mathbb{R}$, denoted as $P_{G}$. We denote the observed sample from the distribution $P_G$ to be $G_{obs}$. Throughout, we assume the network G is formed prior to and independently of treatment assignment Z; in particular, treatment does not alter link formation or interaction patterns among units.
We define the potential outcome function $Y:\{0,1\}^N \times \{0,1\}^{N \times N} \rightarrow \mathbb{R}^N$ as a real-valued function that takes treatment assignment and network adjacency matrix as arguments. We will consider $Y$ to be fixed, and each element $i$ of an image in $Y(Z, G)$, $Y_i(Z, G)$, denotes the outcome of the $i^{th}$ unit in the population under the treatment assignment $Z$ and network assignment $G$. We denote the corresponding potential outcome as $Y_{obs}$ that is, $Y_{obs} = Y(Z_{obs}, G_{obs})$. We note that the potential outcome function implicitly contains information in $X$ and can be written as $Y_{X}(Z, G)$. For the remainder of the article, we will use the notation $Y(Z, G)$ for readability unless deemed necessary. Our framework accounts for interference through the network structure and the treatment assignments of neighboring units. Treatment propagation and outcome-mediated interference are outside the scope of the current framework.

\subsection{Definitions}
The assignment mechanism defines the joint probability distribution of treatment assignment vectors given covariates and potential outcomes, denoted by $P(Z|X, Y(0), Y(1))$. The following are three key assumptions that are considered in the causal inference literature for assignment mechanism: (i) individualistic assignment, which restricts the treatment assignment of a unit to its covariates and outcomes (ii) probabilistic assignment, which assumes non-degenerate probabilities of treatment assignment to units (iii)  unconfoundedness assumption, which restricts the dependence of the probability assignment of units to only covariates, and not the potential outcome (\citet{imbens2015causal}). 
Throughout this paper, we operate in a design-based experimental framework where the treatment assignment mechanism is known and controlled by the experimenter. Thus, the assignment mechanism is unconfounded. Given a well-defined potential outcome function, the individualistic assignment mechanism assumption constitutes \textit{Stable Unit Treatment Value Assumption} (SUTVA). We argue that the underlying network structure among the population is an implicit form of treatment. Hence, we generalize the treatment to a bivariate treatment assignment vector $(Z, G)$. 
We assume that the joint treatment assignment is unconfounded; that is, 
\begin{equation}
\label{unconfoundeness}
    P(Z,G|X,Y(0),Y(1)) =   P(Z,G|X). 
\end{equation}
\begin{remark}
\label{remark no spillover}
    It should be noted that the graph $G$ \textit{only} maps the interference structure in the population. Any other dependence of the potential outcome function on network characteristics and other confounding factors is captured by the study's covariates. We denote the covariates corresponding to network statistics as $G_X$, which may be dependent on G. Apart from $G_X$, in this article, we restrict the potential outcome function to be covariate-free for the sake of brevity. Thus, Equation \ref{unconfoundeness} can be written as 
    \begin{equation*}
    \begin{aligned}
        P(Z,G|X,Y(0),Y(1)) &=   P(Z,G|G_{X},Y(0),Y(1)),\\
        &= P(Z,G|G_{X}).
    \end{aligned}
    \end{equation*}
    For further details on covariate balancing in randomized experiments, refer to \citet{liu2020regression}.
\end{remark}
We formally define the standard sharp null hypothesis of no (direct) treatment effect when the population is not interconnected.
\begin{hypothesis}
    The sharp null hypothesis of no treatment effect: 
    \begin{equation*}
        \forall i \in [N] \quad Y_i(Z) = Y_i(Z') \qquad \forall \quad Z, Z' \in \{0,1\}^N.
    \end{equation*}
\end{hypothesis}
The null hypothesis stated above imposes specific restrictions on the potential outcome function. Extending this to our setup, we get the following definition of the sharp null hypothesis.  
\begin{hypothesis}
    The sharp null hypothesis of no (joint) treatment effect: 
    \begin{equation*}
       \forall i \in [N]\quad Y_i(Z,G) = Y_i(Z',G') \qquad \forall \quad(Z,G), (Z',G') \in \{0,1\}^N\times\{0,1\}^{N\times N} .
    \end{equation*}
\end{hypothesis}
We now define a test statistic that will be used for testing such null hypotheses. Following our setup, we will use the joint treatment assignment, $(Z, G)$, for the rest of the article.  
\begin{definition}
    A test statistic $T$ is a real-valued function of the treatment assignment and $Y_{obs}$ given the covariates $X$, $T_{X}(Z, G, Y_{obs})$. 
\end{definition}
Under the sharp null, given an observed outcome, $Y_{obs}$, we can impute the outcome function $Y$ for all treatment assignments. This is not possible when the null hypothesis is non-sharp. We formally define the non-sharp null hypothesis below. 
\begin{hypothesis}
\label{hyp 3}
    Non-sharp null hypothesis: 
    \begin{equation*}
        \forall i \in [N]\quad Y_i(Z,G) = Y_i(Z',G') \qquad \forall \quad(Z,G), (Z',G') \in \textbf{Z}_i\times\textbf{G}_i \subset\{0,1\}^N\times\{0,1\}^{N\times N}.
    \end{equation*}
\end{hypothesis}
Here, the subset restriction, denoted by $\textbf{Z}\times\textbf{G}:= \textbf{Z}_i\times\textbf{G}_i$, on treatment assignment is chosen by the null hypothesis of interest. We elaborate further using the testing null of no spillover effects, as shown in the example below. Consider the null hypothesis of no spillover effects if the distance between two units is at least k. We assume this setting to be graph-covariate free. 
\\
\begin{example}
\label{ex 1}
    (Null Hypothesis:~no spillover effects at distance k and beyond)
    Let the $k-$distance function $D^{k}_{G}(i) \in \{0,1\}^{N}$ for a unit $i$ in graph $G$ be defined as below:
    \begin{equation*}
        \big[D^{k}_{G}(i)\big]_j = 
        \begin{cases}
        1 & \text{if dist(i,j)}\leq k,\\
        0 & \text{otherwise}.
        \end{cases}
    \end{equation*}
    \begin{equation*}
    \begin{aligned}
         \forall i \in [N]\quad Y_i(Z,G) = Y_i(Z',G') \qquad &\forall \quad(Z,G), (Z',G') \in \{0,1\}^N\times\{0,1\}^{N\times N} \quad, \\
        &\text{s.t.}\quad Z\cdot D_{G}^{(k-1)}(i) = Z'\cdot D_{G'}^{(k-1)}(i).
    \end{aligned}
    \end{equation*}    
\end{example}
 Note the non-imputability of the potential outcome function under the null considered above for different treatment assignments. Precisely, given $(Z_{obs}, G_{obs})$, we can compute potential outcomes for treatment assignments in the restriction set defined above. The null is not sharp since the restriction set is a strict subset of the treatment assignment space. 
 
\subsection{Null hypothesis of no spillover and conditional randomization test}
In this paper, we focus on testing for first-order spillover effects. We take $k=1$ in Example \ref{ex 1} and state the null hypothesis of no spillover effects below.
\begin{hypothesis}
\label{hyp 4}

\begin{mdframed}

    (Null hypothesis of no spillover effects)
    \begin{equation*}
        \begin{aligned}
         \forall i \in [N]\quad {Y_i}_{G_{X_i}}(Z, G) = {Y_i}_{G'_{X_i}}(Z', G') \quad &\forall (Z, G), (Z', G') \in \{0,1\}^N\times\{0,1\}^{N\times N} \quad, \\
        &\text{s.t.}\quad Z_i = Z_i', {G_X}_i={G'_X}_i.
    \end{aligned}
    \end{equation*}
\end{mdframed}
\end{hypothesis}

The treatment $G$ is unrestricted here, with the graph covariates fixed. 
Under the null of no spillover effects, outcomes are invariant to the network structure corresponding to the spillover effect, so G is free to vary. Conditioning on the direct treatment assignment vector fixes the direct effect, which is independent of G. Since graph covariates can induce heterogeneity in the potential outcome function, we also condition on them to capture only spillover effects. This framework enables randomization over the network-generating process in our testing procedure. Since the network is itself a random variable, the natural null is that outcomes are unaffected by both neighbors' treatment assignments and the realized network structure (equivalently, the effective treatment, known as exposure mappings in the literature, is unaffected by the realized network structure).

We restate the null hypothesis of no spillover effects, as described in Hypothesis \ref{hyp 4}.
\begin{example}
\label{hyp 5}
(Null hypothesis of no additive spillover effects)
    \begin{equation*}
    \begin{aligned}
        Y_{i}(Z,G) &= Y^1(Z_i) + Y^2(Z\cdot G_i),\\
                   &= Y^1(Z_i) + Y^2(Z_{\mathcal{N}_{G}(i)}).
    \end{aligned}       
    \end{equation*}
    Using the above assumption of a linearly separable joint treatment effect, we get
    \begin{equation*}
        \begin{aligned}
        \forall i \in [N]\quad Y^2(Z_{\mathcal{N}_{G}(i)}) = Y^2(Z'_{\mathcal{N}_{G'}(i)}) \qquad &\forall \quad(Z,G), (Z',G') \in \{0,1\}^N\times\{0,1\}^{N\times N} \quad, \\
        &\text{s.t.}\quad Z_i = Z_i', {G_X}_i={G'_X}_i.
    \end{aligned}
    \end{equation*}
\end{example}
 This approach to defining non-sharp nulls of spillovers in the literature uses \textit{exposure functions} (\citet{aronow_estimating_2017}). The exposure functions are functions of the treatment assignment vector, $Z$, such that restriction sets can be created on $Z$ by partitioning $Z$ on the range of the exposure function. We can define an exposure function in the above Example \ref{ex 1}: 
\begin{equation*}
    f_{G, i}(Z) = \frac{Z\cdot D_{G}^{(k-1)}(i)}{|D_{G}^{(k-1)}(i)|}.
\end{equation*}

We obtain the setup described in the literature for a fixed $G$, using the above-described exposure function as the restriction set. The setup assumes a spillover mechanism via an exposure function. Since we consider the graph underlying the population to be random, we define restrictions on the joint treatment assignment in Hypothesis \ref{hyp 3}. A spillover mechanism, if it exists, will be captured in the potential outcome function. We describe the currently proposed method in the literature of a conditional-randomization test procedure to obtain valid p-values using `focal units' (\citet{puelz2022graph}).
\begin{proced}
\label{procedure 0}
Consider $F$ as a subset of units and treatment assignment $Z$ such that a test statistic $T(Z|F)$ is imputable under the null.
\begin{enumerate}
    \item Draw $Z_{obs} \sim P(Z)$ and obtain $Y_{obs}$.
    \item Draw $F \sim P(F|Z_{obs})$ and compute $T(Z_{obs}|F)$.
    \item Compute p-value as $P_{Z|F}\left(T(Z|F) > T(Z_{obs}|F)\right)$.
\end{enumerate}
\end{proced}
The above-stated procedure holds finite-sample validity. The conditioning mechanism $F$, termed as focal units, is drawn from $P(F|Z_{obs})$, which the analyst should determine. We consider the random selection method as the conditioning mechanism to obtain the focal units, as proposed in \cite{athey2018exact}. Procedure \ref{procedure 0} tests for deviation from the null hypothesis of no spillover effect on focal units. This ensures the test statistic $T(Z|F)$ is imputable. We require the test statistic to be imputable to obtain its sampling distribution under the null, even though we observe only one treatment assignment vector. Below, we define the imputability of a test statistic for the joint treatment assignment. 

\begin{definition}
    Consider a test statistic $T(Z, G, Y_{obs})$. A test statistic is called imputable under $\mathcal{H}_0$ if 
    \begin{equation*}
    \begin{aligned}    
        T(Z,G,Y(Z,G)) &= T(Z,G,Y(Z',G'))\\
        &\quad P(Z,G|(Z,G)\in \textbf{Z}\times\textbf{G}\text{ },\mathcal{H}_0),P(Z',G'|(Z',G')\in \textbf{Z}\times\textbf{G}\text{ }, \mathcal{H}_0) >0.
    \end{aligned}
    \end{equation*}
\end{definition}
The choice of an imputable test statistic is not unique, though, and different test statistics can yield different test powers. This is determined by the test statistic's responsiveness to the null hypothesis.

\section{Quasi-randomization test based on network}
\label{quasi-rand-test}

We present the quasi-randomization test for the sharp null of no spillover effects. Here ``sharp" refers to imputability under randomization of G, not to the classical fixed-network sense. Our main idea is to generate the null sampling distribution of a test statistic by randomizing over the network-generating process. Intuitively, the underlying network structure also contains information about spillover effects. This can also be seen in Hypothesis \ref{hyp 4}, where the network variable $G$ is unrestricted under the otherwise non-sharp null, keeping the network covariates fixed. Fixing direct treatment assignment and network covariates leads to the imputability of a test statistic of choice. To strictly generalize the conditional randomization test over focal units (\cite{athey2018exact}), one can randomize over the network distribution of $G_{F^{c}}$, that is, over the non-focal units in the population. One feature of this generalization is that it ensures the network characteristics of the focal units remain fixed. One can, then, also randomize over the joint treatment assignment vector $(Z, G)$ for the non-focal units. This extends state-of-the-art methods to the bivariate treatment assignment vector. We now present our method that generalizes this and randomizes over all the units in the population. The following method focuses on the conditional treatment assignment distribution of $G|Z$. We formally define the conditional test statistic, $T_c$, below and show its imputability under the null Hypothesis \ref{hyp 4}.
\begin{proposition}
\label{Prop impute}
    Define $T_c(G|Z_{obs}) := T(Z_{obs},G,Y_{{G_X}_{obs}}(Z_{obs},G))$. $T_c(G|Z_{obs})$ is imputable under the null hypothesis of no spillover effects in Hypothesis \ref{hyp 4} and is equal to $T(Z_{obs}, G, Y_{obs})$. Here, $Y_{obs}=Y_{{G_X}_{obs}}(Z_{obs},G_{obs})$.
\end{proposition}
\begin{proof}
    The proof is given in Appendix \ref{proof impute}
\end{proof}
We state our proposed quasi-randomization testing procedure below. The procedure assumes the knowledge of the network-generating process $P(G|G_{X})$. An elaborate consideration of the network distribution will be presented in Section \ref{quasi permutation test}. Here, we present Procedure \ref{procedure 1} as the oracle procedure, where $P(G)$ is known. We later present Procedure \ref{procedure 2}, which relaxes this by conditioning on the degree sequence. Procedure \ref{procedure 3} further generalizes to ERGMs for cluster-randomized settings.
\begin{proced}
\label{procedure 1}
 Let $(Z_{obs}, G_{obs}, Y_{obs})$ be the observed joint treatment assignment and its corresponding outcome. Consider a test statistic $T(Z, G, Y_{obs})$. 

\begin{enumerate}
    \item Define a new test statistic $T_c(G|Z_{obs}) := T(Z_{obs},G,Y_{obs})$.
    \item Impute the observed test statistic value ${T_c}_{obs} = T_c(G_{obs}|Z_{obs})$.
    \item Consider the randomization distribution $P(G|Z,G_{X})$.
    \item Compute $pval(Z_{obs},G_{obs},Y_{obs}):= \mathbb{E}_{G|Z_{obs}}[\mathcal{I}(T_c(G|Z_{obs})>{T_c}_{obs})|Z_{obs}]$.\footnote{When $T_c$ has a discrete distribution, ties occur with positive probability and the p-value is defined to include the observed graph, which guarantees finite-sample validity without a continuity assumption. In the empirical (sampling-based) versions of the test, this corresponds to the convention $pval = (1 + b)/(1 + B)$, where $B$ is the number of sampled graphs and $b$ the number with $T_c(G \mid Z_{\mathrm{obs}}) > T_c(G_{\mathrm{obs}} \mid Z_{\mathrm{obs}})$.}
\end{enumerate}
\end{proced}
\begin{theorem}
\label{validity rand}
Consider the null hypothesis, $\mathcal{H}_0$, in Hypothesis \ref{hyp 4}. Let $(Z_{obs}, G_{obs}) \sim P(Z, G)$ be the bivariate treatment assignment of a randomized experiment. We assume that $P(Z, G)$ is known. Consider a test statistic $T(Z, G, Y(Z, G))$. Then, Procedure \ref{procedure 1} described above is conditionally valid at level $\alpha \in (0,1)$. That is,
\begin{equation*}
    \mathbb{E}\bigl[\mathcal{I}(\text{p-val}(Z_{obs},G_{obs},Y_{obs})\leq \alpha)|\mathcal{H}_0\bigl]\leq \alpha \qquad \forall \text{ }\alpha \in (0,1).
\end{equation*}
Here, the expectation is for the distribution of $P(G_{obs}|Z_{obs})$. 
\end{theorem}
\begin{proof}
    The proof is given in Appendix \ref{proof validity rand}.
\end{proof}
\begin{remark}
    Note that conditional validity implies unconditional validity of the procedure, too. That is, Procedure \ref{procedure 1} holds finite sample validity with respect to the bivariate joint distribution of the treatment assignment vector $(Z, G)$, following the law of iterated expectations. 
\end{remark}
Theorem \ref{validity rand} proves the finite sample validity of Procedure \ref{procedure 1}. However, since $P(G)$ is unknown in practice, in Section \ref{Permutation test} we develop procedures that avoid direct estimation of the network distribution by conditioning on a sufficient statistic. A crucial component of the proof is identifying the conditional distribution $P(G|Z, G_{X})$ for the validity of Procedure \ref{procedure 1}. This is contingent on the experimental design, and different designs yield different conditional randomization distributions. We compare this with the state-of-the-art methodology presented in Section \ref{Setup}, where the null randomization distribution is derived from $P(Z|F)$. Here, $F$ represents focal units. A common problem in the literature is power loss when implementing methods in cluster-randomized experiments. 

In a cluster-randomized experiment, the population is divided into clusters, and treatments are assigned at the cluster level. If a cluster is treated, all units in that cluster receive the treatment or vice versa. Fixing the treatment assignment of units in the conditioning mechanism might render the null randomization distribution $P(Z|F)$ degenerate, given the restricted span of cluster treatment assignment space. For example, assume there is a focal unit in every cluster. Then, conditioning on the observed treatment of focal units, we obtain $Z|F =\{Z_{obs}\}$ since all the unit treatment assignments can be deterministically imputed by the treatment assignment of the focal unit in that cluster. This reduction in support for the conditional distribution leads to either degeneracy or a substantial loss of power. To address this, our methodology captures the null distribution by randomizing over the conditional distribution $G|(Z, G_{X})$, which has ample support. We discuss this in Sections \ref{cre_section} and \ref{cluster}, where we identify the null randomization distribution under different experimental designs.   

\section{Experimental design}
\label{Experimental design}
Here, we discuss how Procedure \ref{procedure 1} can be implemented in different experimental designs. A change in experimental design corresponds to a change in (direct) treatment assignment mechanism. Thus, correctly identifying the distribution $P(G|Z)$ (and $P(G|Z, G_{X})$) is important to ensure the validity of Procedure \ref{procedure 1}. We derive $P(G|Z)$ for two experimental designs: completely randomized and cluster-randomized. 

\subsection{Completely randomized experiment}
\label{cre_section}
Consider a random sample of units of size $N_t$ from $\mathbb{P}$, which are selected in the treatment group. Let $N_t$ be fixed. Thus, the remaining units are in the control group. Let $N_c = N-N_t$ be the size of the units in the control group.
\begin{equation}
\label{cre}
    P(Z=z\;|\;G) = 
    \begin{cases}
        \frac{1}{\binom{N}{N_t}}& \sum_{i=1}^{N} z_i = N_t,\\
        0 & \text{otherwise}.
    \end{cases}
\end{equation}
We remark that $G \!\perp\!\!\!\perp Z$ by design. Hence, for completely randomized experiments, $P(G|Z) = P(G).$ Since $P(G)$ is not identified by design, we will explore random graph models in Section \ref{sec: random graph}.  
\subsection{Cluster randomized experiment}
\label{cluster}
Consider the following partition of the population $\mathbb{P}$ into clusters $\{C_1,C_2,....,C_k\}$. That is, $ C_i \cap C_j = \emptyset \forall i\neq j \in [k],\ \ 
     \cup_{i=1}^{k} C_i = \mathbb{P}.$
Given such a partition, treatment assignment is performed at the cluster level rather than the unit level. Following \citet{ugander2013graph}, we perform a Bernoulli experiment at the cluster level. Thus, we obtain $W_i \sim Bernoulli(p),\ \
     Z_j = \sum_{i=1}^{k} W_i\cdot \mathcal{I}(j \in C_i).$
Here, $p$ is the probability that a cluster receives treatment, and the experiment is performed on a given cluster partition of the population. Cluster formation is part of the experimental design and is fixed. The experimenter can form these clusters based on the underlying network's structure. Benchmark graph clustering algorithms can be used to obtain a cluster partition that depends on graph characteristics such as modularity, cuts, and related quantities (\citet{newman2006modularity}). We use $\epsilon$-net clustering to obtain the network's cluster partition, as previously considered by \citet{eckles2017design} and \citet{ugander2013graph}. We describe the $\epsilon$-net clustering in detail in Appendix \ref{epsilon net}. While we present a deterministic clustering procedure, it should be noted that any other algorithm, such as a different selection rule or the resulting $\epsilon$-net cluster, constitutes a change in the experimental design. Unlike completely randomized experiments, in cluster randomized experiments, $G \not\!\perp\!\!\!\perp Z$. So,
\begin{equation}
\begin{aligned}
    P(G|Z) &= \frac{P(Z|G)\cdot P(G)}{P(Z)},\\
           &= \frac{\prod_{i=1}^{k} p^{\mathcal{I}(\exists j \in [N] \text{ s.t. } j \in C_i, z_j=1)}\cdot\prod_{i=1}^{k} (1-p)^{\mathcal{I}(\exists j \in [N] \text{ s.t. } j \in C_i, z_j=0)}\cdot P(G)}{P(Z)}.
\end{aligned}
\end{equation}
This shows that identifying the null randomization distribution is challenging in graph cluster randomization experiments. We propose a new approach that uses sufficient statistics from the random graph model to construct a permutation test. 

\section{Permutation test via distributional properties}
\label{quasi permutation test}
Identification of the null randomization distribution $P(G|Z, G_{X})$ involves knowledge of the distribution of the network-generating process, which is not identified by the design of the experiment. To this end, we explore random graph models to model $P(G)$. As discussed in the previous section, identifying the null randomization distribution becomes challenging in cluster-randomized experiments. This challenge is overcome in our proposed methodology, which is a conditional permutation test. We discuss this in Section \ref{Permutation test}. We now describe the random graph models for $P(G)$. 
\subsection{Random graph models}
\label{sec: random graph}
This section describes various random graph models based on which $P(G)$ can be modeled. 
The Erdős-Rényi random graph model has been studied widely, owing to exact computations of desirable graph properties. We present the model formally in Appendix \ref{ER model}. However, the Erdős-Rényi random graph model does not represent most real-world networks, such as social networks, supply-chain networks, and scientific collaboration networks, among many others. We look at a generalization of Erdős-Rényi random graph better at modeling such real-world networks in Section \ref{RG_abit}.  

\subsection{Random graphs with arbitrary degree distribution}
\label{RG_abit}
To better reflect the realistic degree distributions observed in real life, \citet{newman2001random} proposed characterizing a random graph by its degree distribution. We formulate the model below.
\begin{definition}
    Consider a random graph $G \in \{0,1\}^{N\times N}$ such that $P(G) = f(deg(i): i\in \mathbb{P}|p_0,p_1,...,p_{N-1})$. Here, $\{p_i\}_{i=0,...,N-1}$ represent the degree distribution, that is, $p_i$ is the probability that a randomly selected unit has degree $i$ and $f(.)$ represents the probability mass function. We call $G$ a random graph with an arbitrary degree distribution.
\end{definition}

A key feature of random graphs with arbitrary degree distributions is their ability to capture a wide range of degree distributions. This generalizes the Erdős-Rényi random graph model and, at the same time, ensures tractability of important structural graph properties (\citet{newman2002random}). To deploy Procedure \ref{procedure 1}, we must identify $P(G)$. This involves estimating $(p_1,p_2,...,p_N)$ with a single realization of the graph, and we cannot directly implement Procedure \ref{procedure 1} without estimating $P(G)$. For random graphs with arbitrary degree distributions, it is shown that we can estimate the graph's parameters with just one realization, albeit with some uncertainty (\citet{chatterjee2011random}). However, there might be other relational dependencies characterizing the network's formation that are not captured by random graphs with arbitrary degree distributions. For example, the phenomenon of transitivity is not captured in this model: two nodes are more likely to be connected if they have a common neighbor. 

\subsection{Exponential random graph model}
The state-of-the-art method for characterizing the random graph formation in social networks is the exponential random graph model (ERGM; \citet{snijders2006new}). It is a generalization of random graphs with arbitrary degree distributions, in which the distribution of the random graph is characterized by a set of given sufficient statistics. This could also include information on covariates to capture homophily in the network-formation process. In this article, we consider the random graph formation free of covariates. 
\begin{definition}
    Consider a random graph $G \in \{0,1\}^{N\times N}$ such that $P(G=g|\eta)= \frac{1}{\kappa}\cdot{e^{\eta^{T}\cdot s(g)}}$. Here, $s(g)$ represents a vector of sufficient statistics, and $\kappa$ is the normalizing constant. We call $G$ an exponential random graph.
\end{definition}
Common sufficient statistics used in the literature for characterizing social networks include the number of edges, the number of triangles, and the number of stars of different sizes, such as stars of size 2 or 3 (\citet{robins2007introduction}). It is important to carefully choose the model, as an ill-posed ERGM can lead to near-degeneracy (\citet{chatterjee2013estimating, schweinberger2020exponential}). As with random graphs with arbitrary degree distributions, estimating ERGMs is challenging with only a single realization.

Owing to estimation issues with the null distribution, we propose a new method that avoids such estimations and reduces to a simple permutation test. In our proposed methodology, we first consider the sufficient statistics of the random graph distribution as the degree sequence and develop a method conditioning on the sufficient statistic. We then generalize the method to a set of arbitrary sufficient statistics. We exploit the equiprobability of graphs within equivalence classes of $P(G)$ to construct a conditional permutation test, as described below.

\section{Permutation test based on equiprobable events}
\label{Permutation test}
\subsection{Test for a random graph with arbitrary degree distribution}
Here, we present our main result, which proposes obtaining a conditional randomization test without estimating the distribution. Our main idea is to identify a sufficient statistic for the random graph model, which is the degree sequence in the case of random graphs with arbitrary degree distribution, and condition on this sufficient statistic. We sample the null distribution over all the graphs with the same degree sequence as the observed graph. We obtain a permutation test by assuming that the event set of all graphs with the same degree sequence is equiprobable. We now formally present the definitions below and proceed to state and prove the finite sample validity of the proposed method.

We define the degree sequence of a graph $G = (V, E)$.
\begin{definition}
\label{def deg seq}
Let $G=(V,E)$ be a graph with $|V|=N$.  
The degree sequence of $G$, denoted $D$, is the sequence
\[
deg(G) = (\deg_G(v_1),\ldots,\deg_G(v_N))
\]
sorted in nondecreasing order.
\end{definition}
Not all $D \in \mathbb{W}^N$ can be a valid degree sequence. For example, consider $D = \{1,2\}$ as a candidate for a degree sequence for a graph of size $2$. We remark that the maximum degree of a node in the graph is at most $(N-1)$.  Since no vertex in a graph of size two can have a degree of more than 1, $D$ is not a valid degree sequence. Note that graphs with the same degree sequence are not necessarily isomorphic graphs. This implies that the graph's vertices cannot be permuted to obtain the other graph with edges preserved (see Figure \ref{fig:degree_sequence}). 
\begin{figure}[t]
    \centering
    \includegraphics{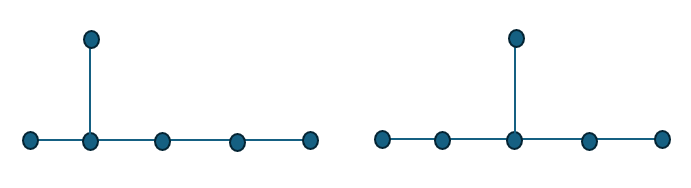}
    \caption{Non-isomorphic graphs with the same degree sequence}
    \label{fig:degree_sequence}
\end{figure}
The converse holds. Consider a set of distinct graphs with the number of units fixed to $N$ and a degree sequence of the graphs fixed to a valid degree sequence $D$. Then, this set of all graphs with degree sequence $D$ will form an equivalence class, and we can partition the space of all random graphs with this equivalence relation. We define this below. 
\begin{definition}
\label{relation}
    Consider a random graph $G \in \{0,1\}^{N\times N} \sim P(G)$ where $P(G) = f(deg(i): i\in V|p_0,p_1,...,p_{N-1})$ where $\{p_i\}_{i=0,...,N-1}$ represent the degree distribution of $G$. Let $\mathcal{G}$ be the support of $P(G)$. We define a relation $d:\mathcal{G}\rightarrow \mathcal{G}$ such that $G \equiv_d G'$ if $deg(G) = deg(G')$ such that $deg_i(G) = deg_i(G') \forall i \in [N]$. 
\end{definition}
It is a quick check that $d$ defines an equivalence relation. Hence, we can partition the support of $P(G)$, where each part represents an equivalence class of $D$.
\begin{proposition}
\label{Prop deg seq}
    Consider $d:\mathcal{G}\rightarrow \mathcal{G}$ such that $G \equiv_d G'$ if $deg(G) = deg(G')$, as defined in Definition \ref{relation}. Then, $d$ is an equivalence relation. 
\end{proposition}
\begin{proof}
    Refer to Appendix \ref{proof deg seq}.
\end{proof}
\begin{proced}
\label{procedure 2}
 Let $(Z_{obs}, G_{obs}, Y_{obs})$ be the joint treatment assignment vector and its corresponding outcome. Consider a test statistic $T_c(G|Z_{obs}) := T(Z_{obs},G,Y_{obs})$.
\begin{enumerate}
    \item Impute the observed value of test statistic ${T_c}_{obs} = T_c(G_{obs}|Z_{obs})$.
    \item Define the equivalence class $[G_{obs}] = \{G \in \mathcal{G}|\quad G \equiv_d G_{obs}\}$. Here, $\equiv_d$ is the relation defined in Definition \ref{relation}.
    \item Compute p-value as $p(Z_{obs},G_{obs},Y_{obs}) := \frac{1}{|[G_{obs}]|}\cdot\sum_{G \in [G_{obs}]}\mathcal{I}(T_c(G|Z_{obs})>T_c(G_{obs}|Z_{obs}))$.
\end{enumerate}
\end{proced}
While Procedure \ref{procedure 2} works well for completely randomized experiments, identifying the null randomization distribution becomes challenging in cluster-randomized settings. We next generalize to this case via the ERGM framework in Procedure \ref{procedure 3}. First, we state and discuss the validity of Procedure \ref{procedure 2}.
\begin{theorem}
\label{thm_permutation}
    Consider the null hypothesis of no spillover effects in Hypothesis \ref{hyp 4}, denoted by $\mathcal{H}_0$. Let $(Z_{obs}, G_{obs}) \sim P(Z, G)$ be the bivariate treatment assignment of a completely randomized experiment. We assume that $G_{obs} \sim P_G$ where $G$ is a random graph with arbitrary degree distribution. Let $G_{X_i} = deg(i)$ where $G_{X_i}$ denote the graph covariate(s) of unit $i$. Consider a test statistic $T(Z, G, Y(Z, G))$. Then, Procedure \ref{procedure 2} described above is conditionally valid at level $\alpha \in (0,1)$. That is,
    \begin{equation*}
    \mathbb{E}\bigl[\mathcal{I}(p(Z_{obs},G_{obs},Y_{obs})\leq \alpha)|\mathcal{H}_0\bigl]\leq \alpha \qquad \forall \text{ }\alpha \in (0,1).
\end{equation*}
Here, the expectation is taken for the distribution of $P(G_{obs}|Z_{obs})$.
\end{theorem}
\begin{proof}
Refer to Appendix \ref{proof_validity_perm}.
\end{proof}
As a corollary, we can substitute any subclass of the degree sequence equivalence class in Procedure \ref{procedure 2}, and the finite-sample validity of the Procedure is retained. For example, we can take the equivalence class to be all the graphs isomorphic to the observed graph, or automorphic to the observed graph, in the second step of Procedure \ref{procedure 2}. We present results for both these Procedures in Section \ref{simulation study}.\\
We emphasize that generating graphs with a given degree sequence can be done in polynomial time. Given fixed parameters like the maximum degree of the graph, one can use state-of-the-art algorithms to efficiently sample graphs with a given degree sequence uniformly from the space of all the graphs with the same degree sequence (\citet{arman2021fast}). It has been shown that the size of this equivalence class can be exponentially large (\citet{barvinok2013number}). It should be noted that one should use a uniform sampler to avoid introducing any bias in the null distribution attributed to sampling. We can also overcome this by effectively taking many samples to approximate the null distribution.\\

\subsection{Test for exponential random graph model}

We generalize the above method to a class of exponential random graph models whose sufficient statistics are permutation-invariant. We can take the graph covariate vector, $G_{X}$, to be the same set of statistics. To sample graphs from $G|G_{X}$, we can condition on the graph automorphism class and generate the null cases from this class. One caveat is that the size of an automorphic graph class may not be large depending on the symmetry of the graph observed (\cite{erdos1963asymmetric}). To overcome this, we explore sampling approximate automorphic graphs. We fix the graph covariate to be the degree of a unit, and sample isomorphic graphs that preserve the degree sequence for the null cases. To achieve more generality, one can specify the graph covariate vector of interest and explore advanced graph sampling or rejection sampling techniques. We formally present the results below.

A graph is considered isomorphic to another graph if a one-to-one correspondence between the vertices of the two graphs preserves the edges between the graphs. Below, we present an equivalence condition for the same. 
\begin{definition}
    Two graphs $G$ and $G'$ are isomorphic if and only if a permutation matrix $P$ exists such that $G = P^{T}G'P$. We denote $G$ and $G'$ as isomorphic graphs using $G \cong G'$. 
\end{definition}
We note that graph isomorphism defines an equivalence relation. The permutation matrix comprises row binary vectors indexed by the one-to-one correspondence between the vertices of the graphs. We now define a class of exponential random graph models based on the permutation-invariance property of their sufficient statistics.  
\begin{definition}
    An exponential random graph model $G \in \{0,1\}^{N\times N}$ is called permutation-invariant if $s(\Pi(g))= s(g)$. Here, $\Pi$ represents the permutation function of the graph's vertices, and $s(.)$ represents the vector of sufficient statistics of the exponential random graph model.
\end{definition}
Many sufficient statistics used in real-life social network modeling follow this property. As discussed, some graph statistics include the number of triangles, the degree sequence, and the number of stars of size 2, and they are permutation-invariant. We model the data-generating process as an exponential random graph with a permutation invariance property and present a permutation testing procedure.
\begin{proced}
\label{procedure 3}   
 Let $(Z_{obs}, G_{obs}, Y_{obs})$ be the joint treatment assignment vector and its corresponding outcome. Consider a test statistic $T_c(G|Z_{obs}) := T(Z_{obs},G,Y_{obs})$.
\begin{enumerate}
    \item Impute the observed value of test statistic ${T_c}_{obs} = T_c(G_{obs}|Z_{obs})$.
    \item Consider the sets $Z^t := \{i\in \mathbb{P}: \: Z_i =t\}\; \forall\; t\in\{0,1\}$.
    \item Define the equivalence class
    \begin{equation*}
        [G_{obs}|Z_{obs}] = \{G \in \mathcal{G}\;|\: G[{Z_{obs}}^0] \cong G_{obs}[{Z_{obs}}^0],\: G[{Z_{obs}}^1] \cong G_{obs}[{Z_{obs}}^1], G \equiv_d G_{obs}\}.
    \end{equation*}   
    Here, $G[Z^t_{obs}]$ denotes the subgraph of $G$ induced by the units in $Z^t_{obs}$, as defined in Section \ref{Setup}.
    \item Compute p-value as 
    \begin{equation*}
        p(Z_{obs},G_{obs},Y_{obs}) := \frac{1}{|[G_{obs}|Z_{obs}]|}\cdot\sum_{G \in [G_{obs}|Z_{obs}]}\mathcal{I}(T_c(G|Z_{obs})>T_c(G_{obs}|Z_{obs})).
    \end{equation*}
\end{enumerate}
\end{proced}
Together, Procedures \ref{procedure 1}, \ref{procedure 2}, and \ref{procedure 3} form a coherent progression from the ideal case of a known network distribution to practical settings with unknown distributions under both completely and cluster-randomized designs. For illustration, we give a simple example of a sample permutation considered in the above Procedure \ref{procedure 3} and then state its validity.
\begin{example}
Consider a line graph comprising 6 units. That is, take $G_{obs}$ as $[1-2-3-4-5-6]$. Consider the units $\{1,2,4\}$ to be treated units, and the remaining to be control units. A valid permutation, according to Step 3 of the Procedure \ref{procedure 3}, permutes the vertices of the same degree and treatment status. In this example, the valid sample permutations are $[1-4-3-2-5-6]$, $[1-4-5-2-3-6]$, and $[1-2-5-4-3-6]$ in addition to $G_{obs}$.
\end{example}
\begin{theorem}
\label{thm_permutation_expo}
    Consider the null hypothesis of no spillover effects in Hypothesis \ref{hyp 4}, denoted by $\mathcal{H}_0$. Let $(Z_{obs}, G_{obs}) \sim P(Z, G)$ be the bivariate treatment assignment of a cluster randomized experiment. We assume that $G_{obs} \sim P_G$ where $G$ is a permutation-invariant exponential random graph model. Let $G_{X_i} = deg(i)$ where $G_{X_i}$ denote the graph covariate(s) of unit $i$. Consider a test statistic $T(Z, G, Y(Z, G))$. Then, Procedure \ref{procedure 3} described above is conditionally valid at level $\alpha \in (0,1)$. That is,
    \begin{equation*}
    \mathbb{E}\bigl[\mathcal{I}(p(Z_{obs},G_{obs},Y_{obs})\leq \alpha)|\mathcal{H}_0\bigl]\leq \alpha \qquad \forall \text{ }\alpha \in (0,1).
\end{equation*}
Here, the expectation is taken for the distribution of $P(G_{obs}|Z_{obs})$.
\end{theorem}
\begin{proof}
Refer to Appendix \ref{proof_validity_perm}.
\end{proof}

\subsection{Test statistic}
\label{test_statistic}
The choice of the test statistic is not unique, and the validity of the proposed test is upheld for any choice of test statistic. Test statistics capture departures from the null distribution, and better-suited test statistics that respond to deviations from the null will improve the test's power. Here, we propose a variation of the Has-Treated-Neighbor test statistic considered in \citet{athey2018exact}. This examines the difference in average treatment effect between control units with at least one treated neighbor and those with only control neighbors. Let $T_{I_c}(G|Z_{obs}) =$
    \begin{equation}  
    \begin{aligned}
        &\frac{\sum_{i=1}^{N}Y_{i}\cdot \mathcal{I}(\sum_{j=1}^{N}G_{ij}\cdot Z_{j_{obs}} >0,Z_{i_{obs}}=0)}{\sum_{i=1}^{N}\mathcal{I}(\sum_{j=1}^{N}G_{ij}\cdot Z_{j_{obs}} >0,Z_{i_{obs}}=0)} - \frac{\sum_{i=1}^{N}Y_{i}\cdot\mathcal{I}(\sum_{j=1}^{N}G_{ij}\cdot Z_{j_{obs}} =0,Z_{i_{obs}} = 0)}{\sum_{i=1}^{N}\mathcal{I}(\sum_{j=1}^{N}G_{ij}\cdot Z_{j_{obs}} =0,Z_{i_{obs}} = 0)}.
    \end{aligned}
    \end{equation}
    
The above test statistic accounts for spillovers from the treated to the control units. Under the null of no spillover effect across all units, regardless of direct treatment status, one can also define the above test statistic for the treated units. Similarly, let $T_{I_t}(G|Z_{obs})$ equals
    \begin{equation}
    \begin{aligned}
        &\frac{\sum_{i=1}^{N}Y_{i}\cdot \mathcal{I}(\sum_{j=1}^{N}G_{ij}\cdot Z_{j_{obs}} >0,Z_{i_{obs}}=1)}{\sum_{i=1}^{N}\mathcal{I}(\sum_{j=1}^{N}G_{ij}\cdot Z_{j_{obs}} >0,Z_{i_{obs}}=1)} - \frac{\sum_{i=1}^{N}Y_{i}\cdot\mathcal{I}(\sum_{j=1}^{N}G_{ij}\cdot Z_{j_{obs}} =0,Z_{i_{obs}} = 1)}{\sum_{i=1}^{N}\mathcal{I}(\sum_{j=1}^{N}G_{ij}\cdot Z_{j_{obs}} =0,Z_{i_{obs}} = 1)}.
    \end{aligned}
    \end{equation}
    We now define the test statistic $T_I(G|Z_{obs})$ for testing the spillover effect across all units. It will be a weighted average of the two test statistics defined above. 
    \begin{equation}
        T_I(G|Z_{obs}) = \frac{N_c}{N}\cdot T_{I_c} + \frac{N_t}{N}\cdot T_{I_t}.
    \end{equation}
    It should be noted that the above test statistics can only be defined if the sub-populations under consideration are non-empty. For example, if no control unit exists with all neighbors under control, or if there are few such units, we face either a definition or a power issue. Thus, we propose a generalization of the above test statistic to improve robustness. We examine the difference in average treatment effect between control units with a proportion of treated units at or above the 0.75 quantile and those with a proportion at or below the 0.25 quantile. Let $T_{quant_c}(G|Z_{obs}) =$
    \begin{equation}
    \begin{aligned}
        &\hspace{-1.5cm}\frac{\sum_{i=1}^{N}Y_{i}\cdot \mathcal{I}(Z_{i_{obs}} =0)\cdot\mathcal{I}(\frac{\sum_{j=1}^{N}G_{ij}\cdot Z_{j_{obs}}}{\sum_{j=1}^{N}G_{ij}} \geq p_{0.75})}{|P_{0.75}|} \\ &- \frac{\sum_{i=1}^{N}Y_{i}\cdot \mathcal{I}(Z_{i_{obs}} =0)\cdot\mathcal{I}(\frac{\sum_{j=1}^{N}G_{ij}\cdot Z_{j_{obs}}}{\sum_{j=1}^{N}G_{ij}} \leq p_{0.25})}{|P_{0.25}|}.
    \end{aligned}
    \end{equation}
    Here, $p_{0.25}$ and $p_{0.75}$ represent the first and the third quartile of the vector 
    \begin{equation*}
        P:= \left\{ \frac{\sum_{j=1}^{N}G_{ij}\cdot Z_j}{\sum_{j=1}^{N}G_{ij}} : Z_i = 0\quad\forall i \in [N]\right\}.
    \end{equation*} 
$P_{0.25}$ and $P_{0.75}$ represent the set of elements in $P$ in the first and third quartile, respectively. In a similar fashion to the test statistic $T_I(G|Z_{obs})$, we define the above for treated units and take a weighted average for the two test statistics to obtain $T_{quant}(G|Z_{obs})$.

We note that the test statistic's power can be improved by using appropriate covariates or nodal attributes. This can be incorporated by considering regression-adjusted test statistics taking residuals for the potential outcome (\citet{rosenbaum_observational_2002}). One can also build test statistics based on models for interference structures. For example, parametric test statistics can be constructed of the linear-in-means model of interference proposed by \citet{manski1993identification}. 
We use the test statistic proposed by \citet{bond201261}, which is the difference between the average over all edges where the neighbor is a treated or control unit. 

\section{Simulation study}
\label{simulation study}
We present Monte Carlo simulations to validate our proposed testing procedure. We consider two test statistics: one based on the quantile of the proportion of treated neighbors in the simulations and another weighted by treated neighbors (first proposed in \citet{bond201261}), as discussed in Section \ref{test_statistic}. The following subsection will describe the data-generating process and the general setup. We conduct a simulation study for two experimental designs: completely randomized and cluster randomized. As a robustness check, we repeat the study for two network-generating processes.

\subsection{Data generating mechanism}
We now present the data-generating process considered in the simulations. Our potential outcome function is linearly separable into direct treatment effect and spillover effect. The spillover mechanism depends on the unit's neighbors' direct treatment assignments. The potential outcome function also has an additive heterogeneous effect coming from a network-based property, free of a direct treatment assignment vector. We propose that the mechanism be contingent on the proportion of treated neighbors. 
\begin{equation}
\label{outcome_prop}
Y_{i}(z) =
\begin{cases}
    \tau_{\text{direct}} \cdot z_{i} + \tau_{\text{spill}} \cdot \frac{\sum_{j=1}^{N}G_{ij} \cdot z_{j}}{\sum_{j=1}^{N}G_{ij}} + \beta_{\text{deg}}\cdot \frac{\sum_{j=1}^{N}G_{ij}}{\text{max}_i{\sum_{j=1}^{N}G_{ij}}}+ \epsilon_{i} & \sum_{j=1}^{N}G_{ij} > 0,  \\
    \tau_{\text{direct}} \cdot z_{i} + \epsilon_{i} & \text{otherwise}.
\end{cases}
\end{equation}

Here, $\tau_{direct}$ captures the direct treatment effect, $\tau_{spill}$ captures the spillover effect, and $\beta_{deg}$ captures the network-effect. The error term $\epsilon_i \sim N(0,1)$ represents the baseline outcome level in the absence of treatment.
Without loss of generality, we take the mean of the normal distribution to be zero, but it can be arbitrarily chosen.
We first present the setup for a completely randomized experiment. The treatment assignment vector is chosen by randomly sampling $N_t$ units out of $N$ units and assigning them treatment (that is, $Z_i = 1$ for all $i$ corresponding to units in the random sample). We state that $N_t$ is fixed and pre-determined by the experimenter. In the performed simulations, we have $599$ total units and $300$ treated units in the model in one case. We study the setup for two generating processes: small-world networks and the stochastic block model. Both models exhibit a high degree of clustering, characterized by the clustering coefficient. The clustering coefficient is the number of connections within a randomly chosen node's neighborhood relative to the total number of possible connections. Consider a regular graph in which all nodes have the same number of connections. Let us denote the number of connections a unit has in a regular graph by $K$. Such graphs have a high clustering coefficient, but the distance between nodes is also high. The average distance between two nodes in real-life networks is low (\citet{amaral2000classes}). To obtain a graph with a high degree of clustering but a low average path length, one can rewire the edges of a $K$-regular graph with some probability, say $p$. This parametrization characterizes a small-world network. The higher the rewiring probability, the lower the average path length. Another class of random graph models that achieve the same property of a high clustering coefficient and a low average path length is the stochastic block model. Consider a graph of size $N$ with a partition of nodes. The size of the parts are $\{n_1,n_2,...,n_K\}$. These parts represent clusters in the graph, and we can model a random graph using them by defining the probability of an edge within a part, between two parts, or between nodes within different parts. Hence, we obtain a probability (symmetric) matrix of size $K\times K$. Note that because we also want a high degree of community structure in the graph, the diagonal elements of the probability matrix are typically higher than the off-diagonal elements. The probability matrix is also called the preference matrix in the literature. We also have a coefficient of network dependence in the potential outcome model, denoted by $\beta_{\text{deg}}$. We take $\beta_{\text{deg}} = 0$ in Table \ref{table: cre} for a comparative study with the state-of-the-art literature (\cite{athey2018exact}). We now define the setup for the simulation study:
\begin{enumerate}
    \item Network:~We first present results taking network distribution as a small-world network with the following parameters:~K (initial connection of all units) = 10 and $p_{rw}$ (probability of rewiring) = 0.1. We then consider a stochastic block model with block sizes $\{50, 100, 40, 110, 299 \}$. The preference matrix, which specifies the probability of an edge between two groups, is given below:
    \begin{table}[h!]
  \centering
  \begin{tabular}{ccccc}
    \hline\vspace{1mm}
    0.08 & 0.01 & 0.01 & 0.01 & 0.01 \\
         & 0.05 & 0.01 & 0.01 & 0.01 \\
         &      & 0.05 & 0.01 & 0.01 \\
         &      &      & 0.05 & 0.01 \\
         &      &      &      & 0.09 \\
    \hline
  \end{tabular}
  \label{pref_matrix}
\end{table}
    \item Direct treatment effect:~We take $\tau_{direct}$ equal to 0 and 4 in our simulations.
    \item Spillover effect:~We take the spillover effect $\tau_{spill}$ to be 0 and 0.4.
    \item Test statistic:~We consider the $T_{quant}$ test statistic described in Section \ref{test_statistic}. It takes the difference in the treatment effect of control units in the first quartile of the proportion of treated neighbors to control units in the last quartile. We define the same for treated units and take the weighted average. We also consider the $T_{bond}$ test statistic, proposed in \citet{bond201261}, defined as follows:
    \begin{equation*}
        T_{bond}(G|Z_{obs}) = \frac{\sum_{i,j\in [N]}G_{ij}\cdot Z_j\cdot Y_i}{\sum_{i,j\in [N]}G_{ij}\cdot Z_j} - \frac{\sum_{i,j\in [N]}G_{ij}\cdot (1-Z_j)\cdot Y_i}{\sum_{i,j\in [N]}G_{ij}\cdot (1-Z_j)}.
    \end{equation*}
\end{enumerate}

The null distribution was generated by sampling 1000 graphs with the same (labeled) degree sequence as the observed graph and degree-preserving isomorphic to it, in a second option. 
{
\begin{table}[H]
\caption{Rejection rates for the null hypothesis of no spillover effects for a completely randomized experiment (CRE). 
For Procedure \ref{procedure 3}, one does not necessarily need to permute treated and control units separately in the case of CRE, unlike in cluster-randomized experiments.}
  \label{table: cre}
  \centering
  \vspace{0.25cm}
  \small
  \begin{tabular}{cccccccc}
    \hline
    \multirow{2}{*}{\shortstack{Network\\ Distribution}}&\multirow{2}{*}{\shortstack{Direct\\Effect}} & \multirow{2}{*}{\shortstack{Spillover\\ Effect}} & \multicolumn{2}{c}{\shortstack{\rule{0pt}{2.5ex}\small Degree \\ Permutation Test \\ (Procedure \ref{procedure 2})}}  & \multicolumn{2}{c}{\shortstack{\rule{0pt}{2.5ex}\small Isomorphism \\ Permutation Test \\(Procedure \ref{procedure 3})}} & \shortstack{\rule{0pt}{2.5ex}\small Conditional \\ Focal Test \\ (Procedure \ref{procedure 0})}\\
    \cline{4-8}
    & &  & $T_{bond}$ & $T_{quant}$&  $T_{bond}$ &$T_{quant}$ & $T_{elc}$\\
    \hline
    \multirow{4}{*}{\shortstack{Small world \\network}}& 0 & 0 & 0.045 & 0.049 & 0.051 & 0.045 & 0.050\\[0.1cm]
    & 4 & 0 & 0.054 & 0.046 & 0.046 & 0.050 & 0.046\\[0.1cm]
    & 0 & 0.4 & 0.187 & 0.148 & 0.193 & 0.163 & 0.112\footnotemark[1]\\[0.1cm]
    & \textbf{4} & \textbf{0.4} & \textbf{0.074} & \textbf{0.151} & \textbf{0.071} & \textbf{0.156} & \textbf{0.060}\footnotemark[2]\\
    \hline
    \multirow{4}{*}{\shortstack{Stochastic block \\model}}& 0 & 0 & 0.045 & 0.062 & 0.047 & 0.046 & 0.054\\[0.1cm]
    & 4 &  0 & 0.050 & 0.055 & 0.053 & 0.056 & 0.052\\[0.1cm]
    & 0 &  0.4 & 0.191 & 0.217 & 0.200 & 0.206 & 0.065\\[0.1cm]
    & \textbf{4} &  \textbf{0.4} & \textbf{0.064} & \textbf{0.220} & \textbf{0.055} & \textbf{0.214} & \textbf{0.048}\\
    \hline
  \end{tabular}
\end{table}
}

\footnotetext[1]{We reported the number presented in \citet{athey2018exact}. We obtained a rejection rate of 0.073 in our replication.}
\footnotetext[2]{We obtain the same rejection rate in our replication as reported in \citet{athey2015exact}.}
The p-value is the proportion of sampled graph realizations for which the imputed test statistic exceeds the observed test statistic. We reject the null hypothesis if the p-value obtained is less than the chosen significance level of 0.05. We replicated the above-described setup 4,000 times and calculated the proportion of times we rejected the null hypothesis.  We present the rejection rates in Table \ref{table: cre}. 
For robustness, we run simulations using two network-generating processes. We obtain the rejection rates for the same setup for the procedure proposed by \citet{athey2018exact}. We implement random focal-vertex selection and present results for the Edge-Level-Contrast test statistic (\citet{athey2018exact}). We label this test as `Conditional Focal Test'.

We then run simulations for a cluster-randomized experimental design. We present the results in Table \ref{table: cluster}. The total number of units in the simulation is the same as before: 599. We obtain the clusters in the graph using the epsilon-net clustering algorithm described in Algorithm 1 (Appendix \ref{epsilon net}). We set epsilon equal to 3. This implies that all clusters contain units within three hops of the center, whereas distinct cluster centers are at least three hops apart. We assume the probability of treating a cluster is 0.5. We implemented a Bernoulli trial at the cluster level. Under Procedure \ref{procedure 3}, the null distribution is approximated by taking 1000 samples of graphs isomorphic to the original graph, such that the corresponding subgraphs on treated units and controlled units remain isomorphic and the labeled degree sequence remains preserved. This is obtained by considering permutations in which treated units can be mapped only to treated units (thus, control units can be mapped only to control units), and mapping occurs only between two units with the same degree. We replicated this setup 4000 times and calculated the proportion of times the proposed test rejected the null. We label Procedure \ref{procedure 3} as `Block Isomorphism Permutation Test' in Table \ref{table: cluster}. This was performed for the two network-generating processes as before. 
{
 \begin{table}[H]
\caption{Rejection rates for the null hypothesis of no spillover effects for a cluster randomized experiment. 
\cite{athey2018exact} and \cite{puelz2022graph} are omitted as both yield degenerate null distributions under cluster randomization, as discussed in Section \ref{sec: comparision}.}
  \label{table: cluster}
  \centering
  \vspace{0.25cm}
  \begin{tabular}{ccccc}
    \hline
    \multirow{2}{*}{\shortstack{Network \\Distribution}}&\multirow{2}{*}{\shortstack{Direct \\Effect}} & \multirow{2}{*}{\shortstack{Spillover\\ Effect}} & \multicolumn{2}{c}{\shortstack{\rule{0pt}{2.5ex} Block Isomorphism \\ Permutation Test\\ (Procedure \ref{procedure 3})}} \\
    \cline{4-5}\rule{0pt}{2.5ex}
    & &  & {\rule{0pt}{2.5ex}$T_{bond}$} & {\rule{0pt}{2.5ex}$T_{quant}$} \\[0.1cm]
    \hline
    \multirow{4}{*}{\shortstack{Small world\\ network}}& 0 & 0 & 0.053 & 0.046\\[0.1cm]
    & 4 & 0 & 0.051 & 0.054\\[0.1cm]
    & 0 & 0.4 & 0.420 & 0.360\\[0.1cm]
    & 4 & 0.4 & 0.435 & 0.360\\
    \hline
    \multirow{4}{*}{\shortstack{Stochastic block\\ model}}& 0 & 0 & 0.046 & 0.050\\[0.1cm]
    & 4 &  0 & 0.046 & 0.052\\[0.1cm]
    & 0 &  0.4 & 0.197 & 0.206\\[0.1cm]
    & 4 &  0.4 & 0.211 & 0.214\\
    \hline
  \end{tabular}
\end{table}
}
{
\begin{table}[H]
\caption{Rejection rates for the null hypothesis of no spillover effects for the outcome model including graph covariates. Here, $\beta_{deg} = 0.4$, and the Procedures deployed are the degree-preserving isomorphism/block isomorphism permutation test for a completely/cluster-randomized experiment.}
  \label{table: degree positive rep}
  \centering
  \vspace{0.25cm}
  \small
  \begin{tabular}{ccccccc}
    \hline
    \multirow{2}{*}{\shortstack{Network\\ Distribution}}&\multirow{2}{*}{\shortstack{Direct\\Effect}} & \multirow{2}{*}{\shortstack{Spillover\\ Effect}} & \multicolumn{2}{c}{\shortstack{\rule{0pt}{2.5ex}\small Completely \\ Randomized Experiment}}  & \multicolumn{2}{c}{\shortstack{\rule{0pt}{2.5ex}\small Cluster \\ Randomized Experiment}} \\
    \cline{4-7}
    & &  & $T_{bond}$ & $T_{quant}$&  $T_{bond}$ &$T_{quant}$ \\
    \hline
    \multirow{4}{*}{\shortstack{Small world \\network}}& 0 & 0 & 0.052 & 0.045 & 0.048 & 0.046 \\[0.1cm]
    & 4 & 0 & 0.050 & 0.056 & 0.054 & 0.053 \\[0.1cm]
    & 0 & 0.4 & 0.184 & 0.162 & 0.421 & 0.370 \\[0.1cm]
    & 4 & 0.4 & 0.065 & 0.168 & 0.415 & 0.330 \\
    \hline
    \multirow{4}{*}{\shortstack{Stochastic block \\model}}& 0 & 0 & 0.050 & 0.048 & 0.050 & 0.048 \\[0.1cm]
    & 4 &  0 & 0.051 & 0.050 & 0.046 & 0.043 \\[0.1cm]
    & 0 &  0.4 & 0.197 & 0.216 & 0.210 & 0.205 \\[0.1cm]
    & 4 &  0.4 & 0.071 & 0.208 & 0.186 & 0.202 \\
    \hline
  \end{tabular}
\end{table}
}

For robustness, we also perform simulations for the setup in which the potential outcome function depends on the graph, without direct treatment assignment. For this, we set $\beta_{\text{deg}}=0.4$, inducing a positive correlation with the degree of the unit and its corresponding potential outcome. We present the results in Table \ref{table: degree positive rep}. The simulations are replicated for both the completely randomized and the cluster-randomized experiments. The procedures deployed are the degree-preserving isomorphism permutation test and the degree-preserving block isomorphism permutation test. The rest of the setup remains the same as before.

\subsection{Comparison with related work}
\label{sec: comparision}
We see that our proposed methods have type 1 error controlled appropriately. This corresponds to rows with $\tau_{spill}$ equal to zero in Tables \ref{table: cre}, \ref{table: cluster}, and \ref{table: degree positive rep}. When $\tau_{spill}$ equals 0.4, our method performs competitively with existing methods in the literature (\citet{athey2018exact}). We highlight that $\tau_{direct}$ is 4 and $\tau_{spill}$ is 0.4, where our method significantly improves. The significance of the setup lies in its resemblance to many naturally occurring settings, where the direct effect is meaningfully more than the spillover effect (for example, see \citet{cai2015social}). We see in Table \ref{table: cre} that our proposed method increases power as the clustering of the network-generating process increases, as in the stochastic block model. This contrasts with the method proposed by \citet{athey2018exact}, which does not retain power. In Appendix \ref{insights_bond}, we add more insights on the performance of the $T_{bond}$ test statistic as observed in Table \ref{table: cre}.

An interesting result is that, in cluster-randomized experiments, greater clustering is associated with lower power. This can be observed by comparing the power values of the small-world network to those of the stochastic block model. This may be attributed to a reduction in the sample space of imputed values of the test statistic. Hence, a higher degree of clustering increases cluster sizes, thereby restricting direct treatment assignment. For instance, the observed graph in the small-world network had 16 clusters, compared to 11 in the stochastic block model.

We also conduct a simulation study for cluster-randomized experiments using the methods proposed in \citet{athey2018exact} and \citet{puelz2022graph}. As discussed in Section \ref{quasi-rand-test}, we encounter the degenerate null distribution and do not present the corresponding results. This is more pronounced in the method proposed by \citet{athey2018exact} than in that of \citet{puelz2022graph}. Since \citet{puelz2022graph} considers the experiment's design when selecting focal units, we see more robustness of the biclique test (\citet{puelz2022graph}) where degenerate null distributions are reduced in frequency. Overall, our method is more robust to degeneracy and performs well in terms of power in our simulations. We also see that the $T_{bond}$ test statistic performs slightly better than the $T_{quant}$ test statistic in cluster-randomized setups with small-world networks. This may be attributed to $T_{bond}$ being weighted by treated counts in the graph rather than the corresponding quantile-based weights in $T_{quant}$, and to the heterogeneity of the clustered assignment setup. We emphasize that $T_{bond}$ in conjunction with our methodology continues to be valid regardless of the size of the direct effect. The method proposed by \citet{bond201261} does not work for settings with non-zero direct effects. We refer the reader to \citet{athey2018exact} for a detailed discussion. Our method also retains the same power with graphical confoundedness of the potential outcome function, as shown in Table \ref{table: degree positive rep} where $\beta_{\text{deg}}$ = 0.4. 

\section{Weather insurance adoption in rural China}
\label{illustration}
We reanalyze a field experiment conducted in rural China among rice farmers that examined neighborhood interference effects (cited in \citet{cai2015social}). The research motive behind the experiment was to determine whether social networks among farmers influence weather insurance adoption. The experiment involved intensive training sessions to promote crop insurance products to farmers against extreme weather. A primary goal of the experiment was to determine whether information dissemination occurs within the farmers' community and, if so, what mechanisms drive these exchanges. 
A deeper understanding of insurance product adoption via social network mechanisms helps devise a more cost-effective government policy campaign to enhance product adoption among the target community. An important reason to study interference in such a setup is to avoid biased estimates and accurately evaluate policy performance measures.


We discuss the experimental design setup by the researchers. Researchers hypothesize that greater access to information about weather insurance products would increase their adoption among rice farmers. To capture access to information differentially as treatment and control, two types of information sessions were created: simple sessions of 20 minutes, which provided an introduction to the insurance contract, and intensive sessions of 45 minutes that covered the insurance contract in much more detail, including explaining to them what the benefits are of taking up insurance. To capture information dissemination across the social network, these sessions were conducted in two rounds, with a 3-day gap between them. The idea behind this setup is to allow the farmers time to discuss the insurance policy and to identify the primary mechanism of information diffusion by examining whether households in round 2 with more friends who received intensive sessions in round 1 exhibit a higher insurance take-up rate. 

Second-round participants were further randomized into three groups. After the round 2 sessions, the first group was given no additional information, the second group was told about the overall take-up rate in the previous round, and the third group was told in detail about who purchased the product and who did not. The randomization of the second group was used to identify the primary mechanism of information diffusion. For illustration purposes, we do not analyze this part of the experiment.
Although the social network among farmers is not experimentally randomized, we treat it as a realization of an underlying network-generating process, consistent with the interpretation in Section \ref{Setup}. Social ties among rural farmers evolve over time through repeated interactions and shared experiences, making the random graph assumption plausible in this setting. 

The randomization design is stratified by `Household size' and `Area of rice production per capita'. Procedure \ref{procedure 3} extends naturally to stratified randomization. In stratified randomization, treatment is assigned independently within each stratum, with a fixed number of treated units per stratum. The block isomorphism permutation test conditions on treatment status, permuting only units of the same treatment status and degree. If strata membership is observed, one can further restrict permutations to units within the same stratum, and the uniformity argument in the proof of Theorem \ref{thm_permutation_expo} carries through within each stratum independently. Since treatment is assigned independently across strata, validity follows from independence across strata.

There were 4902 households across 47 different villages in the experiment. Social networks across villages are assumed to be independent. Each household is randomized among the 4 different groups:~(i) Receiving a simple session in Round 1, (ii) Receiving an intensive session in Round 1, (iii) Receiving a simple session in Round 2, (iv) Receiving an intensive session in Round 2. After the session ends, the households decide to purchase insurance. This is modeled as a binary outcome variable, with 1 indicating a purchase and 0 indicating no purchase. Under the assumption that no information diffusion can occur between participant households in Round 1 because of the immediate purchase decision the household had to make, we can assume there are no edges among the households in Round 1. 

In our analysis, we treat households receiving simple sessions as the control group and those receiving intensive sessions as the treated group. 
Since first-order spillovers can flow from intensively treated Round-1 households to their Round-2 neighbors, we restrict attention to the network structure linking these two rounds. This includes the network structure among households in Rounds 1 and 2 only. Our analysis was run on 433 households of randomly chosen villages (labeled Type II in the experiment), of which 222 were in the treated group. 

We ran our proposed Procedure \ref{procedure 3}, the block isomorphism test, and approximated the null distribution by taking 10,000 samples of graphs with the equivalence class defined in Step 3 of Procedure \ref{procedure 3}. We obtained a p-value of 0.0126. Because this is below the conventional 0.05 threshold, we conclude that there are significant spillover effects.\footnote{To the best of our knowledge, strata information is not provided in the data for the Type II villages. We therefore treat the data as from a completely randomized experiment, which is more conservative than stratified randomization for variance estimation \citep{miratrix2013adjusting}.} In other words, households that receive the intensive information session appear to transmit the information to neighboring households that were not directly treated. Such spillovers have important implications for inference. If they are ignored, the estimated average treatment effect may be biased toward zero. As a result, one might incorrectly conclude that the intervention has no meaningful impact when, in fact, part of its effect is simply diffusing beyond the treated group. Properly accounting for the spillover mechanism is therefore essential for accurately interpreting treatment effects in this setting. 

\section{Discussion}
\label{conclusion}

We develop a new methodology to test for the non-sharp null hypothesis of no spillover effect in an experimental setup. By treating the network as a random variable rather than a fixed quantity, our method constructs a quasi-randomization test that delivers finite-sample valid inference while offering substantial power improvements over existing approaches, particularly in cluster-randomized trials. This opens a new and computationally tractable path for interference testing, while leaving several directions for future work.

Since the null distribution is sampled via quasi-randomization of the network-generating process, a natural extension is to study these procedures in observational studies. Incorporating covariates or nodal attribute information to better characterize the random graph null model may enhance the test's power and remains an open problem. Relaxing the assumption that the network is formed independently of treatment to accommodate settings with endogenous network formation is another important direction. 


As noted, the power gains come at the cost of a modeling assumption on the network-generating process, so analyzing robustness under partial information and network misspecification is another important direction. In particular, while exponential random graph models fit many social networks well, they can be ill-suited to settings with distinct generative structures, such as physically driven networks (e.g., atmospheric mixing) or temporally evolving networks (e.g., citation networks). Extending the method to alternative network-generating processes suited to the setting at hand is promising.

\vspace{0.25cm}
{\small
\textbf{Acknowledgments.} We thank Avi Feller, Fabrizia Mealli, Fredrik Sävje, Panos Toulis, Qingyuan Zhao, ACIC 2024 conference participants, and the Tom Ten Have Award committee for their encouragement and feedback. We also thank Sai Sriramya Gorripati, Parshuram Hotkar, Sumit Kunnumkal, and ISB seminar participants for their helpful comments. We gratefully acknowledge Hemanth Kumar and the ISB Institute of Data Science for providing access to the computing cluster. The IMS 2024 ICSDS Student Travel Award partially supported Supriya Tiwari’s research. Pallavi Basu’s research is partially supported by the SERB MATRICS award MTR/2022/000073.}

\bibliographystyle{plainnat}  
{\small \bibliography{references} }

@article{liu2020regression,
  title={Regression-adjusted average treatment effect estimates in stratified randomized experiments},
  author={Liu, Hanzhong and Yang, Yuehan},
  journal={Biometrika},
  volume={107},
  number={4},
  pages={935--948},
  year={2020},
  publisher={Oxford University Press}
}

@article{athey2018exact,
  title={Exact p-values for network interference},
  author={Athey, Susan and Eckles, Dean and Imbens, Guido W},
  journal={Journal of the American Statistical Association},
  volume={113},
  number={521},
  pages={230--240},
  year={2018},
  publisher={Taylor \& Francis}
}

@article{aronow_estimating_2017,
	title = {Estimating average causal effects under general interference, with application to a social network experiment},
	volume = {11},
	issn = {1932-6157},
	doi = {10.1214/16-AOAS1005},
	language = {en},
	number = {4},
	urldate = {2022-08-29},
	journal = {The Annals of Applied Statistics},
	author = {Aronow, Peter M and Samii, Cyrus},
	month = dec,
	year = {2017}
}

@article{puelz2022graph,
  title={A graph-theoretic approach to randomization tests of causal effects under general interference},
  author={Puelz, David and Basse, Guillaume and Feller, Avi and Toulis, Panos},
  journal={Journal of the Royal Statistical Society Series B: Statistical Methodology},
  volume={84},
  number={1},
  pages={174--204},
  year={2022},
  publisher={Oxford University Press}
}

@article{cox1958planning,
  title={Planning of experiments},
  author={Cox, David R},
  isbn={9780471181866},
  journal={Wiley Publication in Applied Statistics},
  year={1958},
  publisher={John Wiley}
}

@article{cai2015social,
  title={Social networks and the decision to insure},
  author={Cai, Jing and Janvry, Alain De and Sadoulet, Elisabeth},
  journal={American Economic Journal: Applied Economics},
  volume={7},
  number={2},
  pages={81--108},
  year={2015},
  publisher={American Economic Association 2014 Broadway, Suite 305, Nashville, TN 37203-2425}
}

@article{snijders2006new,
  title={New specifications for exponential random graph models},
  author={Snijders, Tom AB and Pattison, Philippa E and Robins, Garry L and Handcock, Mark S},
  journal={Sociological methodology},
  volume={36},
  number={1},
  pages={99--153},
  year={2006},
  publisher={Wiley Online Library}
}

@article{chatterjee2013estimating,
  title={Estimating and understanding exponential random graph models},
  author={Chatterjee, Sourav and Diaconis, Persi},
  journal={The Annals of Statistics},
  volume={41},
  number={5},
  pages={2428--2461},
  year={2013},
  publisher={Institute of Mathematical Statistics}
}

@article{bond201261,
  title={A 61-million-person experiment in social influence and political mobilization},
  author={Bond, Robert M and Fariss, Christopher J and Jones, Jason J and Kramer, Adam DI and Marlow, Cameron and Settle, Jaime E and Fowler, James H},
  journal={Nature},
  volume={489},
  number={7415},
  pages={295--298},
  year={2012},
  publisher={Nature Publishing Group UK London}
}

@article{schweinberger2020exponential,
  title={Exponential-Family Models of Random Graphs},
  author={Schweinberger, Michael and Krivitsky, Pavel N and Butts, Carter T and Stewart, Jonathan R},
  journal={Statistical Science},
  volume={35},
  number={4},
  pages={627--662},
  year={2020},
  publisher={JSTOR}
}

@article{pouget2019testing,
  title={Testing for arbitrary interference on experimentation platforms},
  author={Pouget-Abadie, Jean and Saint-Jacques, Guillaume and Saveski, Martin and Duan, Weitao and Ghosh, Souvik and Xu, Ya and Airoldi, Edoardo M},
  journal={Biometrika},
  volume={106},
  number={4},
  pages={929--940},
  year={2019},
  publisher={Oxford University Press}
}

@article{sobel2006randomized,
  title={What do randomized studies of housing mobility demonstrate? Causal inference in the face of interference},
  author={Sobel, Michael E},
  journal={Journal of the American Statistical Association},
  volume={101},
  number={476},
  pages={1398--1407},
  year={2006},
  publisher={Taylor \& Francis}
}

@article{fisher1936design,
  title={Design of experiments},
  author={Fisher, Ronald Aylmer},
  journal={British Medical Journal},
  volume={1},
  number={3923},
  pages={554},
  year={1936},
  publisher={BMJ Publishing Group}
}

@article{newman2002random,
  title={Random graph models of social networks},
  author={Newman, Mark EJ and Watts, Duncan J and Strogatz, Steven H},
  journal={Proceedings of the National Academy of Sciences},
  volume={99},
  number={suppl\_1},
  pages={2566--2572},
  year={2002},
  publisher={National Acad Sciences}
}

@article{robins2007introduction,
  title={An introduction to exponential random graph (p*) models for social networks},
  author={Robins, Garry and Pattison, Pip and Kalish, Yuval and Lusher, Dean},
  journal={Social Networks},
  volume={29},
  number={2},
  pages={173--191},
  year={2007},
  publisher={Elsevier}
}

@article{fosdick2018configuring,
  title={Configuring random graph models with fixed degree sequences},
  author={Fosdick, Bailey K and Larremore, Daniel B and Nishimura, Joel and Ugander, Johan},
  journal={SIAM Review},
  volume={60},
  number={2},
  pages={315--355},
  year={2018},
  publisher={SIAM}
}

@article{chatterjee2011random,
  title={Random graphs with a given degree sequence},
  author={Chatterjee, Sourav and Diaconis, Persi and Sly, Allan},
  journal={The Annals of Applied Probability},
  pages={1400--1435},
  year={2011},
  publisher={JSTOR}
}

@inproceedings{blake2014marketplace,
  title={Why marketplace experimentation is harder than it seems: The role of test-control interference},
  author={Blake, Thomas and Coey, Dominic},
  booktitle={Proceedings of the fifteenth ACM Conference on Economics and Computation},
  pages={567--582},
  year={2014}
}

@article{li2022random,
  title={Random graph asymptotics for treatment effect estimation under network interference},
  author={Li, Shuangning and Wager, Stefan},
  journal={The Annals of Statistics},
  volume={50},
  number={4},
  pages={2334--2358},
  year={2022},
  publisher={Institute of Mathematical Statistics}
}

@article{basse2019randomization,
  title={Randomization tests for peer effects in group formation experiments},
  author={Basse, Guillaume and Ding, Peng and Feller, Avi and Toulis, Panos},
  journal={arXiv preprint arXiv:1904.02308},
  year={2019}
}

@article{leung2020treatment,
  title={Treatment and spillover effects under network interference},
  author={Leung, Michael P},
  journal={Review of Economics and Statistics},
  volume={102},
  number={2},
  pages={368--380},
  year={2020},
  publisher={MIT Press One Rogers Street, Cambridge, MA 02142-1209, USA journals-info~…}
}

@article{newman2001random,
  title={Random graphs with arbitrary degree distributions and their applications},
  author={Newman, Mark EJ and Strogatz, Steven H and Watts, Duncan J},
  journal={Physical Review E},
  volume={64},
  number={2},
  pages={026118},
  year={2001},
  publisher={APS}
}

@article{manski2013identification,
  title={Identification of treatment response with social interactions},
  author={Manski, Charles F},
  journal={The Econometrics Journal},
  volume={16},
  number={1},
  pages={S1--S23},
  year={2013},
  publisher={Oxford University Press Oxford, UK}
}

@article{manski1993identification,
  title={Identification of endogenous social effects: The reflection problem},
  author={Manski, Charles F},
  journal={The Review of Economic Studies},
  volume={60},
  number={3},
  pages={531--542},
  year={1993},
  publisher={Wiley-Blackwell}
}

@inproceedings{toulis2013estimation,
  title={Estimation of causal peer influence effects},
  author={Toulis, Panos and Kao, Edward},
  booktitle={International Conference on Machine Learning},
  pages={1489--1497},
  year={2013},
  organization={PMLR}
}

@article{bowers2013reasoning,
  title={Reasoning about interference between units: A general framework},
  author={Bowers, Jake and Fredrickson, Mark M and Panagopoulos, Costas},
  journal={Political Analysis},
  volume={21},
  number={1},
  pages={97--124},
  year={2013},
  publisher={Cambridge University Press}
}

@article{blume2015linear,
  title={Linear social interactions models},
  author={Blume, Lawrence E and Brock, William A and Durlauf, Steven N and Jayaraman, Rajshri},
  journal={Journal of Political Economy},
  volume={123},
  number={2},
  pages={444--496},
  year={2015},
  publisher={University of Chicago Press Chicago, IL}
}

@inproceedings{viger2005efficient,
  title={Efficient and simple generation of random simple connected graphs with prescribed degree sequence},
  author={Viger, Fabien and Latapy, Matthieu},
  booktitle={International Computing and Combinatorics Conference},
  pages={440--449},
  year={2005},
  organization={Springer}
}

@article{milenkovic2009optimized,
  title={Optimized null model for protein structure networks},
  author={Milenkovi{\'c}, Tijana and Filippis, Ioannis and Lappe, Michael and Pr{\v{z}}ulj, Nata{\v{s}}a},
  journal={PLoS One},
  volume={4},
  number={6},
  pages={e5967},
  year={2009},
  publisher={Public Library of Science San Francisco, USA}
}

@article{sah2014exploring,
  title={Exploring community structure in biological networks with random graphs},
  author={Sah, Pratha and Singh, Lisa O and Clauset, Aaron and Bansal, Shweta},
  journal={BMC Bioinformatics},
  volume={15},
  number={1},
  pages={1--14},
  year={2014},
  publisher={BioMed Central}
}

@book{imbens2015causal,
  title={Causal Inference in Statistics, Social, and Biomedical Sciences},
  author={Imbens, Guido W and Rubin, Donald B},
  year={2015},
  publisher={Cambridge University Press}
}

@inproceedings{ugander2013graph,
  title={Graph cluster randomization: Network exposure to multiple universes},
  author={Ugander, Johan and Karrer, Brian and Backstrom, Lars and Kleinberg, Jon},
  booktitle={Proceedings of the 19th ACM SIGKDD International Conference on Knowledge Discovery and Data Mining},
  pages={329--337},
  year={2013}
}

@article{newman2006modularity,
  title={Modularity and community structure in networks},
  author={Newman, Mark EJ},
  journal={Proceedings of the National Academy of Sciences},
  volume={103},
  number={23},
  pages={8577--8582},
  year={2006},
  publisher={National Acad Sciences}
}

@article{eckles2017design,
  title={Design and analysis of experiments in networks: Reducing bias from interference},
  author={Eckles, Dean and Karrer, Brian and Ugander, Johan},
  journal={Journal of Causal Inference},
  volume={5},
  number={1},
  pages={20150021},
  year={2017},
  publisher={De Gruyter}
}

@book{rosenbaum_observational_2002,
	series = {Springer Series in Statistics},
	title = {Observational Studies},
	isbn = {978-0-387-98967-9},
	publisher = {Springer},
	author = {Rosenbaum, Paul R},
	year = {2002},
	lccn = {2001049264},
}

@article{arman2021fast,
  title={Fast uniform generation of random graphs with given degree sequences},
  author={Arman, Andrii and Gao, Pu and Wormald, Nicholas},
  journal={Random Structures \& Algorithms},
  volume={59},
  number={3},
  pages={291--314},
  year={2021},
  publisher={Wiley Online Library}
}

@article{barvinok2013number,
  title={The number of graphs and a random graph with a given degree sequence},
  author={Barvinok, Alexander and Hartigan, John A},
  journal={Random Structures \& Algorithms},
  volume={42},
  number={3},
  pages={301--348},
  year={2013},
  publisher={Wiley Online Library}
}

@article{amaral2000classes,
  title={Classes of small-world networks},
  author={Amaral, Lu{\i}s A Nunes and Scala, Antonio and Barthelemy, Marc and Stanley, H Eugene},
  journal={Proceedings of the National Academy of Sciences},
  volume={97},
  number={21},
  pages={11149--11152},
  year={2000},
  publisher={National Acad Sciences}
}

@article{athey2015exact,
  title={Exact P-Values for Network Interference},
  author={Athey, Susan and Eckles, Dean and Imbens, Guido W},
  journal={NBER Working Paper},
  volume={w21313},
  year={2015}
}

@article{erdHos1960evolution,
  title={On the evolution of random graphs},
  author={Erd{\H{o}}s, Paul and R{\'e}nyi, Alfr{\'e}d and others},
  journal={Publ. Math. Inst. Hung. Acad. Sci},
  volume={5},
  number={1},
  pages={17--60},
  year={1960}
}

@article{watts1998collective,
  title={Collective dynamics of ‘small-world’ networks},
  author={Watts, Duncan J and Strogatz, Steven H},
  journal={Nature},
  volume={393},
  number={6684},
  pages={440--442},
  year={1998},
  publisher={Nature Publishing Group}
}

@article{zhong2024unconditional,
  title={Unconditional randomization tests for interference},
  author={Zhong, Liang},
  journal={arXiv preprint arXiv:2409.09243},
  year={2024}
}

@article{zhang2023randomization,
  title={What is a randomization test?},
  author={Zhang, Yao and Zhao, Qingyuan},
  journal={Journal of the American Statistical Association},
  volume={118},
  number={544},
  pages={2928--2942},
  year={2023},
  publisher={Taylor \& Francis}
}

@article{aronow2012general,
  title={A general method for detecting interference between units in randomized experiments},
  author={Aronow, Peter M},
  journal={Sociological Methods \& Research},
  volume={41},
  number={1},
  pages={3--16},
  year={2012},
  publisher={SAGE Publications Sage CA: Los Angeles, CA}
}

@article{miratrix2013adjusting,
  title={Adjusting treatment effect estimates by post-stratification in randomized experiments},
  author={Miratrix, Luke W and Sekhon, Jasjeet S and Yu, Bin},
  journal={Journal of the Royal Statistical Society Series B: Statistical Methodology},
  volume={75},
  number={2},
  pages={369--396},
  year={2013},
  publisher={Oxford University Press}
}

@article{erdos1963asymmetric,
  title={Asymmetric graphs},
  author={Erd{\H{o}}s, Paul and R{\'e}nyi, Alfr{\'e}d},
  journal={Acta Math. Acad. Sci. Hungar},
  volume={14},
  number={295-315},
  pages={3},
  year={1963}
}

@article{savje2021average,
  title={Average treatment effects in the presence of unknown interference},
  author={S{\"a}vje, Fredrik and Aronow, Peter and Hudgens, Michael},
  journal={Annals of Statistics},
  volume={49},
  number={2},
  pages={673},
  year={2021}
}

\newpage
\begin{appendices}

\appendix  
\begin{center}
    {\Huge\appendixname}
\end{center}

\section{Algorithm:~$\epsilon$-net cluster}
\label{epsilon net}
An $\epsilon$-net in the graph is a set of units of the graph such that any two units in the set are at least $\epsilon$ distance away from each other, and any unit outside the set is within $\epsilon$ distance. Here, distance represents the shortest path length between two units in the graph. To form an $\epsilon$-net cluster, one can pick a vertex and remove all the vertices within $\epsilon$ distance of the vertex. Then, pick another vertex and iteratively perform the same process until you exhaust all the vertices. The resultant set will be an $\epsilon$-net. Ties in the degree sort are broken uniformly at random and independently of the treatment assignment, so that the clustering is permutation-equivariant in distribution. We present an algorithm below to obtain an $\epsilon$-net cluster.\\[10pt]
\begin{algorithm}[http]
\label{algo:1}
  \caption{$\epsilon$-net cluster}
  \textbf{Input:} 
  \begin{itemize}
    \item Given network $G_{obs}$
  \end{itemize}

  \textbf{Procedure:} \\
  Sort $G_{obs}$ by degree of vertices in ascending order\\
  $S$ is equal to $G_{obs}$\\
  \While{$S$ is not empty}{
    \begin{itemize}
      \item Step 1: Pick the last vertex in $S$, denoted by $v$, and store in $V$
      \item Step 2: Find all units within $\epsilon$ distance of $v$
      \item Step 3: Store the units found in Step 2 as a cluster and remove these units from $S$
    \end{itemize}
  }
    Store the resultant cluster outputs in $C$\\
    Permute the outputs to the original ordering and store them in $C, V$\\
  \textbf{Output:} C and V 
\end{algorithm}

\section{Proofs}
\subsection{Proof of Proposition \ref{Prop impute}}
\label{proof impute}
\begin{proof}
    Note that under $\mathcal{H}_0$, 
    \begin{equation*}
        Y_{{G_X}_i}(Z,G) = Y_{{G'_X}_i}(Z',G') \qquad Z_i=Z'_i, {G_X}_i={G'_X}_i.
    \end{equation*}
    Thus, 
    \begin{equation*}
    \begin{aligned}
        T_c(G|Z_{obs}) &= T(Z_{obs},G,Y_{{G_X}_{obs}}(Z_{obs},G)),\\
        &= T(Z_{obs},G,Y_{{G_X}_{obs}}(Z_{obs},G')) \qquad P(G'|Z_{obs},{G_X}_{obs})>0\text{ under }\mathcal{H}_0.\\ 
    \end{aligned}
    \end{equation*}
    Since $Z_{obs}$ is arbitrary, we conclude imputability of the $T_c(G|Z_{obs})$ under the null $\mathcal{H}_0$. Take $G'=G_{obs}$, we get $T_c(G|Z_{obs})= T(Z_{obs},G,Y_{obs})$. 
\end{proof}
\subsection{Proof of Proposition \ref{Prop deg seq}}
\label{proof deg seq}
\begin{proof}
    $G\equiv_dG \quad \forall G \in \mathcal{G}$. Hence, $d$ is reflexive. If for any distinct $G$ and $G'$ in $\mathcal{G}$, $G\equiv_dG'\implies deg(G) = deg(G') \implies deg(G')=deg(G)\implies G'\equiv_d G$. Hence, $d$ is symmetric. Similarly, for any three distinct $G,G'$ and $G''$ in $\mathcal{G}$, $G\equiv_dG', G'\equiv_dG''\implies deg(G)=deg(G')$ and $deg(G')=deg(G'')\implies deg(G)=deg(G'')\implies G\equiv_dG''$. Hence, $d$ is transitive.
\end{proof}

\subsection{Proof of Theorem \ref{validity rand}}
\label{proof validity rand}
\begin{proof}
    We define p-value as
    \begin{equation*}
        p_{Z_{obs}} = P(T_c(G|Z_{obs})> T_c(G_{obs}|Z_{obs})| Z_{obs},{G_X}_{obs}).
    \end{equation*}
Given $Z_{obs}$ and ${G_X}_{obs}$, we define a variable $U$ with the same distribution as $T_c(G|Z_{obs})$. We emphasize that $U$ is a univariate distribution induced by $P(G|Z_{obs}, {G_X}_{obs})$. Thus,
\begin{equation*}
    p_{Z_{obs}} = 1 - F_{U}(G_{obs}|Z_{obs},{G_X}_{obs}).
\end{equation*}
Here, $F_U(.)$ represents the cumulative distribution of $U$. Then,
\begin{equation*}
    T_c(G|Z_{obs}) \stackrel{d}{=} T_c(G_{obs}|Z_{obs})\qquad \text{given } \mathcal{H}_0.
\end{equation*}
Thus, under $\mathcal{H}_0$, the distribution of $U$ is $P(G_{obs}|Z_{obs},{G_X}_{obs})$ where the randomness is induced by $G_{obs}\sim P(G|Z_{obs},{G_X}_{obs})$ and we get 
\begin{equation*}
     p_{Z_{obs}} = 1 - F_{U}(U).
\end{equation*}
Hence, by using the probability integral transform theorem, we obtain 
\begin{equation*}
    P(p_{Z_{obs}} \leq \alpha) \leq \alpha.
\end{equation*}
\end{proof}

\subsection{Proofs of Theorem \ref{thm_permutation} and Theorem \ref{thm_permutation_expo}}
\label{proof_validity_perm}
\begin{proof}
    We define p-value as
    \begin{equation*}
    \begin{aligned}
        p_{Z_{obs},G_{obs}} &= P(T_c(G|Z_{obs})> T_c(G_{obs}|Z_{obs})| Z = Z_{obs}, G \in [G_{obs}]),\\
        p_{Z_{obs}} &= P(T_c(G|Z_{obs})> T_c(G_{obs}|Z_{obs})| Z = Z_{obs}).\\
    \end{aligned}
    \end{equation*}
Given $Z_{obs}$ and $[G_{obs}]$, we define a variable $U$ with the same distribution as $T_c(G|Z_{obs})$. We emphasize that $U$ is a univariate distribution as induced by $P(G|Z=Z_{obs},G \in [G_{obs}])$. Thus,
\begin{equation*}
    p_{Z_{obs},G_{obs}} = 1 - F_{U}(G_{obs}|Z = Z_{obs},G\in [G_{obs}]).
\end{equation*}
Here, $F_U(.)$ represents the cumulative distribution of $U$. Then,
\begin{equation*}
    T_c(G|Z_{obs}) \stackrel{d}{=} T_c(G_{obs}|Z_{obs})\qquad \text{given } \mathcal{H}_0.
\end{equation*}
Thus, under $\mathcal{H}_0$, the distribution of $U$ is $P(G_{obs}|Z_{obs})$ where the randomness is induced by $G_{obs}\sim P(G|Z = Z_{obs},G\in[G_{obs}])$ and we get 
\begin{equation*}
     p_{Z_{obs},G_{obs}} = 1 - F_{U}(U).
\end{equation*}
Hence, by using the probability integral transform theorem, we obtain 
\begin{equation*}
    P(p_{Z_{obs},G_{obs}} \leq \alpha) \leq \alpha.
\end{equation*}
Let $P_1,P_2,...,P_d$ be the partition induced by $\equiv_d$. Pick $G_1,G_2,...,G_d$ such that $G_i \in P_i$ for all $i=1,2,...,d$. Hence,
\begin{equation*}
\begin{aligned}
     P(p_{Z_{obs},G_i} \leq \alpha) &\leq \alpha\qquad i\in [d],\\
     \sum_{i\in [d]}P(p_{Z_{obs},G_i}\leq \alpha)\cdot P(G \in [G_i])) &\leq \sum_{i\in [d]}\alpha\cdot P(G \in [G_i]),\\
     P(p_{Z_{obs}} \leq \alpha) &\leq \alpha.
\end{aligned}    
\end{equation*}
Let $[G_{obs}] = \{G \in \{0,1\}^{N \times N}:\: deg(G)=deg(G_{obs}) \: s.t.\: deg_i(G)= deg_i(G_{obs})\; \forall i \in [N]\}$. We argue that the distribution of $P(G|Z=Z_{obs}, G\in [G_{obs}])$ is a uniform distribution under a completely randomized experiment. Consider $g,g' \in [G_{obs}]$.
\begin{equation}
\label{eq cond com}
     P(G = g|Z=Z_{obs},G \in [G_{obs}]) = \frac{P(Z=Z_{obs}|G = g,G \in [G_{obs}])\cdot P(G = g| G \in [G_{obs}])}{P(Z= Z_{obs}|G \in [G_{obs}])}.
\end{equation}
\begin{equation}
\label{eq cond 1 com}
\begin{aligned}
     P(G = g|G \in [G_{obs}]) &=\frac{f_{p_0,p_1,...,p_{N-1}}(deg(g))}{\sum_{G\in[G_{obs}]} f_{p_0,p_1,...,p_{N-1}}(deg(G))},\\
     &= \frac{f_{p_0,p_1,...,p_{N-1}}(deg(g'))}{\sum_{G\in[G_{obs}]} f_{p_0,p_1,...,p_{N-1}}(deg(G))},\\
     &= P(G = g'|G \in [G_{obs}]).
\end{aligned}
\end{equation}
\begin{equation}
\label{eq cond 2 com}
\begin{aligned}
    P(Z=Z_{obs}|G = g,G \in [G_{obs}])  &= \frac{1}{\binom{N}{\sum Z_{obs, i}}},\\
    &= P(Z=Z_{obs}|G = g',G \in [G_{obs}]).
\end{aligned}
\end{equation}

Using Equation \eqref{eq cond 1 com} and \eqref{eq cond 2 com} in Equation \eqref{eq cond com}, we obtain 
\begin{equation*}
    P(G = g|Z=Z_{obs},G \in [G_{obs}]) = P(G = g'|Z=Z_{obs},G \in [G_{obs}]).
\end{equation*}
Since $g$ and $g'$ are arbitrary, we conclude that the distribution of $G$ conditional on the event $\{Z=Z_{obs},G\in [G_{obs}]\}$ is uniform under completely randomized experiment. Therefore,
\begin{equation*}
    p_{Z_{obs},G_{obs}} = \frac{1}{|[G_{obs}]|}\cdot\sum_{G \in [G_{obs}]}\mathcal{I}(T_c(G|Z_{obs})>T_c(G_{obs}|Z_{obs})).
\end{equation*}
Hence, we prove that Procedure \ref{procedure 2} is valid. We now prove the validity of Procedure \ref{procedure 3}. Let 
\[
[G_{obs}|Z_{obs}] = \{G \in [G_{obs}]\:|\: G[{Z_{obs}}^0] \cong G_{obs}[{Z_{obs}}^0],\: G[{Z_{obs}}^1] \cong G_{obs}[{Z_{obs}}^1] \}.
\]
Here, ${Z_{obs}}^t = \{i\in \mathbb{P}: \: {Z_{obs}}_i =t\}$ for $t\in\{0,1\}$. Note that $[G_{obs}|Z_{obs}]\subseteq [G_{obs}]$, where $[G_{obs}]=\{G \in \mathcal{G}| G \cong G_{obs}\}$. We argue that the distribution of $P(G|Z=Z_{obs}, G\in [G_{obs}|Z_{obs}])$ is a uniform distribution under a cluster randomized experiment. Let $C$ be a random variable defined as $\mathcal{A}\circ G$. Here, $\mathcal{A}$ represents the clustering algorithm in Algorithm \ref{algo:1}.

Consider $g,g' \in [G_{obs}|Z_{obs}]$. Then $P(G = g|Z=Z_{obs},G \in [G_{obs}|Z_{obs}]) =$
\begin{equation}
\label{eq cond' com}
     \frac{P(Z=Z_{obs}|G = g,G \in [G_{obs}|Z_{obs}])\cdot P(G = g| G \in [G_{obs}|Z_{obs}])}{P(Z= Z_{obs}|G \in [G_{obs}|Z_{obs}])}.
\end{equation}
\begin{equation}
\label{eq cond 1' com}
\begin{aligned}
     P(G = g|G \in [G_{obs}|Z_{obs}]) &=\frac{f_{\eta }(s(g))}{\sum_{G\in[G_{obs}|Z_{obs}]} f_{\eta}(s(G)},\\
     &=\frac{f_{\eta }(s(\pi(G_{obs})))}{\sum_{G\in[G_{obs}|Z_{obs}]} f_{\eta}(s(\pi'(G_{obs}))},\\
     &=\frac{f_{\eta }(s(G_{obs}))}{\sum_{G\in[G_{obs}|Z_{obs}]} f_{\eta}(s(G_{obs})},\\
     &= P(G = g'|G \in [G_{obs}|Z_{obs}]).
\end{aligned}
\end{equation}
\begin{equation}
\label{eq cond 2' com}
\begin{aligned}
    P(Z=Z_{obs}\:|\:G = g,& G \in [G_{obs}|Z_{obs}])\\  
    &\hspace{-3cm}= P(\sum_{i=1}^{|C|} W_i\cdot \mathcal{I}(j \in C_i) = {Z_{obs}}_j\; j\in [N]\:|\:G=g,G\in [G_{obs}|Z_{obs}])\\
    &\hspace{-3cm}= P(\sum_{i=1}^{|\pi(C_{obs})|} W_i\cdot \mathcal{I}(j \in \pi({C_{obs}}_i)) = {Z_{obs}}_j\; j\in [N]\:|\:G=\pi(G_{obs}),G\in [G_{obs}|Z_{obs}])\\
    &\hspace{-3cm}= P(\sum_{i=1}^{|C_{obs}|} W_i\cdot \mathcal{I}(j \in \pi({C_{obs}}_i)) = t\quad\forall j\in {Z_{obs}}^t, t\in \{0,1\}\:|\:G=\pi(G_{obs}),G\in [G_{obs}|Z_{obs}])\\
    &\hspace{-3cm}= P(\sum_{i=1}^{|C_{obs}|} W_i\cdot \mathcal{I}(\pi^{-1}(j) \in {C_{obs}}_i) = t\quad\forall j\in {Z_{obs}}^t, t\in \{0,1\}\:|\:G=\pi(G_{obs}),G\in [G_{obs}|Z_{obs}])\\
    &\hspace{-3cm}= P(\sum_{i=1}^{|C_{obs}|} W_i\cdot \mathcal{I}(j \in {C_{obs}}_i) = t\quad\forall \pi(j)\in {Z_{obs}}^t, t\in \{0,1\}\:|\:G=\pi(G_{obs}),G\in [G_{obs}|Z_{obs}])\\
    &\hspace{-3cm}= P(W=W_{obs}) = P(Z=Z_{obs}|G = g',G \in [G_{obs}|Z_{obs}]).
\end{aligned}
\end{equation}

Using Equation \eqref{eq cond 1' com} and \eqref{eq cond 2' com} in Equation \eqref{eq cond' com}, we obtain 
\begin{equation*}
    P(G = g|Z=Z_{obs},G \in [G_{obs}|Z_{obs}]) = P(G = g'|Z=Z_{obs},G \in [G_{obs}|Z_{obs}]).
\end{equation*}
Since $g$ and $g'$ are arbitrary, we conclude that the distribution of $G$ conditional on the event $\{Z=Z_{obs},G\in [G_{obs}|Z_{obs}]\}$ is uniform under cluster randomized experiment. Therefore,
\begin{equation*}
    p_{Z_{obs},G_{obs}} = \frac{1}{|[G_{obs}|Z_{obs}]|}\cdot\sum_{G \in [G_{obs}|Z_{obs}]}\mathcal{I}(T_c(G|Z_{obs})>T_c(G_{obs}|Z_{obs})).
\end{equation*}
Hence, we prove that Procedure \ref{procedure 3} is valid.

\end{proof}

\section{Additional numerical simulations}
\label{additional_sim}
We present the Erdős-Rényi random graph model and the data-generating process considered in the simulations here. To give an overview, we consider a linearly separable potential outcome model for the outcome function, composed of two effects: a direct treatment effect and a spillover treatment effect. The spillover parameter is weighted by a function of the unit's neighbor. We consider two such outcome models for the simulations: one that captures spillovers through the presence of a treated neighbor, and the other through the proportion of treated neighbors. 

\subsection{Erdős-Rényi random graph model}\label{ER model}
We formally describe the Erdős-Rényi random graph model (\citet{erdHos1960evolution}). The model described below is the probabilist variant of the (Erdős-Rényi) random graph model. 
We fix the number of nodes $N$, and each of the $\binom{N}{2}$ possible edges is present independently with probability $p$. The number of edges $|E|$ is therefore random, with $|E| \sim \mathrm{Binomial}\!\left(\binom{N}{2},\, p\right)$. That is, 
\begin{equation}
\label{eq:ER}
    P(G = G_{\mathrm{obs}} \mid p) = p^{|E|} \, (1-p)^{\binom{N}{2} - |E|}.
\end{equation}
To model $P(G)$ using Equation (\ref{eq:ER}), we have to estimate the parameter $p$. We obtain the MLE of the parameter $p$ and estimate it as $\frac{|E|}{\binom{N}{2}}$.
We conduct an extensive numerical simulation study and verify the finite sample validity of the proposed randomization test. We also obtain the test's power under different scenarios. An Erdős-Rényi random graph model is the simplest network model one can consider, and it may not be representative of the kind of networks we observe in the real world. For example, one observes a high degree of clustering in the graph's social networks (\citet{watts1998collective}). One drawback of the model is the independence and identical probability of formation of two edges in the graph. One way to characterize this is to evaluate the degree distribution of the random graph. It has been shown that the degree distribution of a large  Erdős-Rényi random graph is a Poisson distribution. Most real-life networks violate this property and have highly skewed degree distributions. In the main text, we generalize this and characterize random graphs with an arbitrary degree distribution. 

\subsection{Data generating mechanism}
The following describes the two potential outcome models:
\begin{equation}
\label{outcome1}
    Y_{i}(z) = \tau_{direct}\cdot z_{i} + \tau_{spill}\cdot \mathcal{I}(\sum_{j=1}^{N}G_{ij}\cdot z_j >0) + \epsilon_{i},
\end{equation}
\begin{equation*}
    \epsilon_i \sim \mathcal{N}(0,1) \qquad\forall i \in [N].
\end{equation*}
Here, $\epsilon_i$ represents the baseline level of outcomes in the absence of treatment assignment. We take the mean of the normal distribution to be zero without loss of generality. We now re-state the second potential outcome model from Section \ref{simulation study}.
\begin{equation}
\label{outcome2}
Y_{i}(\mathbf{z}) =
\begin{cases}
    \tau_{\text{direct}} \cdot z_{i} + \tau_{\text{spill}} \cdot \frac{\sum_{j=1}^{N}G_{ij} \cdot z_{j}}{\sum_{j=1}^{N}G_{ij}} + \epsilon_{i} & \sum_{j=1}^{N}G_{ij} > 0,  \\
    \tau_{\text{direct}} \cdot z_{i} + \epsilon_{i} & \text{otherwise}.
\end{cases}
\end{equation}
The treatment assignment vector is chosen by randomly sampling $N_t$ units out of $N$ units and giving them treatment (that is, $Z_i = 1$ for all $i$ corresponding to units in the random sample). Here, $N_t$ is fixed and pre-determined by the experimenter. In the performed simulations, we have 100 total units and 5 treated units in the model in one case. We then extend the study to 500 units, and 250 were treated in the second case.

Consider the following setup:
\begin{enumerate}
    \item Network: We work with the standard Erdős-Rényi random graph model $G(n,p)$ where $n$ is the number of nodes and $p$ is the probability of edge formation between two nodes. In our setup, we take $p = 0.2$.
    \item Direct treatment effect: We take $\tau_{direct} = 4$ in our simulations. We also perform the simulations with $\tau_{direct} = 1$ for robustness. It is noted that $\tau_{direct}$ has no impact on the validity and power of the test.
    \item Spillover effect: We vary the spillover effect, $\tau_{spill}$ from 0 to 5, keeping $\tau_{direct}$ fixed.
    \item Test statistic: We consider the test statistic taking the difference of treatment effect on control units with control neighbors and control units with treated neighbors, $T_{I_c}(G|Z_{obs})$. We also consider the test statistic based on the difference in treatment effect of a low proportion of treated neighbors and a high proportion of treated neighbors, $T_{quant}(G|Z_{obs})$. Refer to Section \ref{test_statistic} for further details.
\end{enumerate}
We generate the randomization distribution of the test statistic under consideration by repeated drawing from the Erdős-Rényi random graph model with the probability of an edge between two nodes being 0.2. The simulations in Table \ref{table:1} were carried out on 100 units, with 5 units treated and the rest in the control group. We approximated the randomization distribution of the test statistic by sampling 200 graphs. We calculate the corresponding p-value by the proportion of times the test statistic imputed was larger than the test statistic for the observed network. If the p-value is less than 0.05, we reject the null hypothesis of no spillover effects. We replicated the above 450 times and calculated rejection rates as the proportion of times the p-value was significant. For this setup, the outcome model \eqref{outcome1} with $\tau_{direct}$ equals to 4 and $\tau_{spill}$ is presented in Table \ref{table:1}, varying from $0.0$ to $0.7$.

We then repeat the simulations for the above setup using the outcome model \eqref{outcome2} and present the results in Table \ref{table:2}. We note that the power of the test decreases substantially for model \eqref{outcome2}, compared with model \eqref{outcome1}. This is due to the simplicity of the outcome model \eqref{outcome1}, and we present its results for verification. Another caveat in the above setting is the nature of the test statistic $T_{I_c}(G|Z_{obs})$. We might not be able to control any control units with all their neighbors under control, depending on the input graph, invalidating the procedure. So, the above setup involves 100 units with 5 units in treatment. We find numerically that the procedure fails when the number of treated units exceeds 5. We present the results for model \eqref{outcome2} with $T_{quant}$ as test statistic in Table \ref{table:2}. Here, the population size is 500, with 250 treated units. $\tau_{direct}$ is 4 and varying $\tau_{spill}$ is used for power performance measure. 

In all the above models, we took the probability of an edge between two nodes as $0.2$. We present the same setup in Table \ref{table:2}, with the probability of an edge between two nodes estimated below. $G_{obs}$ is generated from Erdős-Rényi random graph with $p = 0.2$.
\begin{equation*}
    \hat{p} = \frac{|E_{G_{obs}}|}{\binom{N}{2}}.
\end{equation*}
Here, $|E_{G_{obs}}|$ represents number of edges in the observed graph $G_{obs}$. The results are presented in Table \ref{table:2}. We highlight that type 1 error is controlled at 0.05 across all setups (when $\tau_{spill}$ is zero). When the null hypothesis is false, we observe moderate power. This holds when the network distribution is estimated parametrically.
\newpage
\begin{table}[H]
 \caption{Rejection rates for the null hypothesis of no spillover effects. Outcome model \eqref{outcome1} with test statistic $T_{I_c}$.} 
  \label{table:1}
  \centering  
  \vspace{0.25cm}
  \begin{tabular}{cc}    
    \hline
    $\tau_{spill}$\hspace{0.5cm} & $\text{Rejection Rate}\hspace{0.5cm}$ \\
    \hline
    0.0 & 0.038 \\[0.25cm]
    0.1 & 0.077 \\[0.25cm]
    0.2 & 0.167 \\[0.25cm]
    0.3 & 0.293 \\[0.25cm]
    0.4 &  0.444\\[0.25cm]
    0.5 &  0.620\\[0.25cm]
    0.6 &  0.749\\[0.25cm]
    0.7 & 0.858 \\
    \hline
  \end{tabular}
  
\end{table}
\begin{table}[H]
\caption{Rejection rates for the null hypothesis of no spillover effects. Outcome model \eqref{outcome2} with test statistic $T_{quant}$, and $T_{quant}$ with estimated Erdős-Rényi random graph.} 
  \label{table:2}
  \centering
  \vspace{0.25cm}
  \begin{tabular}{ccc}
    \hline
    \multirow{2}{3em}{$\tau_{spill}$}& \multicolumn{2}{c}{$\text{Rejection Rate}$} \\
    \cline{2-3}
    & $T_{quant}$ & $T_{quant}$ with $\hat{p}$\\
    \hline
    0  & 0.057 & 0.044\\[0.25cm]
    1  & 0.124 &  0.138\\[0.25cm]
    2  & 0.402 & 0.469\\[0.25cm]
    3  & 0.747 & 0.729\\[0.25cm]
    4  & 0.947 & 0.927\\[0.25cm]
    5  & 0.984 & 0.989\\[0.25cm]
    \hline
  \end{tabular}
\end{table}

\section{Performance of test statistic $T_{bond}$ on Procedure \ref{procedure 2}}
\label{insights_bond}
\begin{figure}[H]
    \centering
    \small
    \includegraphics[width=0.8\textwidth]{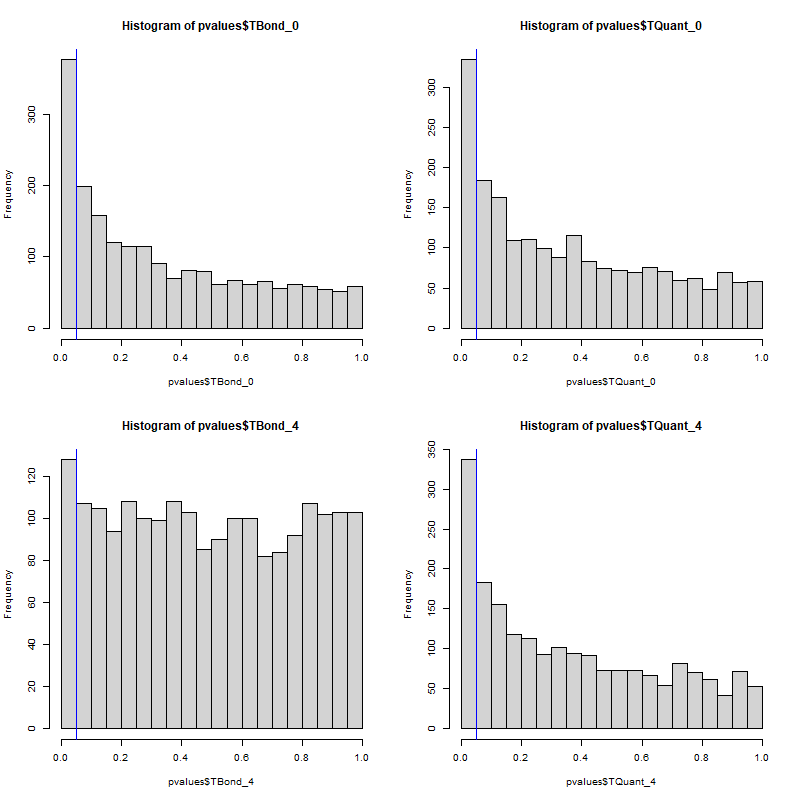}
    \caption{
    P-value distributions under Procedure \ref{procedure 2} over 2,000 replications with spillover effect 0.4. Top row: direct effect 0; bottom row: direct effect 4. Left column: $T_{bond}$; right column: $T_{quant}$. The blue line indicates the 0.05 threshold. The closer-to-uniform distribution of $T_{bond}$ (lower left) explains the power drop observed in Table \ref{table: cre}.}
    \label{fig:hist_pvalues}
\end{figure}
\end{appendices}

\end{document}